\documentclass{nseJournal}
\usepackage{xcolor}
\usepackage{placeins}
\usepackage{amsthm, amsmath}
	\newtheorem{assumption}{Assumption}
	\newtheorem{constraint}{Constraint}
	
	\newtheorem{lemma}{Lemma}
	
	\newtheorem{theorem}{Theorem}
	\newtheorem{conjecture}{Conjecture}

\usepackage{lineno}

\begin{document}

\title{Parallel Approximate Ideal Restriction Multigrid for Solving the S$_N$ Transport Equations} 
\addAuthor{\correspondingAuthor{Joshua Hanophy}}{a}
\correspondingEmail{jthano@tamu.edu}
\addAuthor{Ben S. Southworth}{b}
\addAuthor{Ruipeng Li}{c}
\addAuthor{Jim Morel}{a}
\addAuthor{Tom Manteuffel}{b}

\addAffiliation{a}{Texas A\&M University, Department of Nuclear Engineering\\ 210 Animal Industries, College Station, TX 77843}
\addAffiliation{b}{University of Colorado at Boulder, Department of Applied Mathematics\\ Engineering Center ECOT 225, 526 UCB, Boulder, CO 80309}
\addAffiliation{c}{Lawrence Livermore National Lab\\ 7000 East Avenue, Livermore, CA 94550 }

\addKeyword{transport}
\addKeyword{multigrid}
\addKeyword{sweep}

\titlePage

\begin{abstract}
The computational kernel in solving the $S_N$ transport equations is the parallel sweep, which corresponds to
directly inverting a block lower triangular linear system that arises in discretizations of the linear transport equation. 
Existing parallel sweep algorithms are fairly efficient on structured grids, but still have polynomial scaling, $P^{1/d}$ for
$d$ dimensions and $P$ processors. Moreover, an efficient scalable parallel sweep algorithm for use on general
unstructured meshes remains elusive. Recently, a classical algebraic multigrid (AMG) method based on approximate
ideal restriction (AIR) was developed for nonsymmetric matrices and shown to be an effective solver for linear transport. 
Motivated by the superior scalability of AMG methods (logarithmic in $P$) as well as the simplicity with which AMG
methods can be used in most situations, including on arbitrary unstructured meshes, this paper investigates the use
of parallel AIR (pAIR) for solving the $S_N$ transport equations with source iteration in place of parallel sweeps. 
Results presented in this paper show that pAIR is a robust and scalable solver. Although sweeps are still shown to
be much faster than pAIR on a structured mesh of a unit cube, pAIR is shown to perform similarly on both a structured
and unstructured mesh, and offers a new, simple, black box alternative to parallel transport sweeps.
\end{abstract}

\section{Introduction}
Consider the monoenergetic transport equation with spatially dependent isotropic cross sections,
$\sigma_s$ and $\sigma_t$, and an isotropic source, $q$,
\begin{equation}
\label{eq:mono_transport}
\vec{\Omega} \cdot \nabla \psi \left(  \vec{r}, \vec{\Omega} \right) + \sigma_t \psi \left(  \vec{r}, \vec{\Omega} \right) = 
\frac{\sigma_s}{4\pi}\phi \left(  \vec{r}  \right) + \frac{q}{4\pi},
\end{equation}
where $\psi$ is the angular flux and $\phi$ is the scalar flux equal to the integral of the angular flux over all directions, $\phi = \int_{4\pi}\psi d\Omega$.
The $S_N$ equations are derived by first selecting a quadrature set to approximate the integral over all directions and then defining the scalar flux as
\begin{align}
\begin{split}
\phi \left(  \vec{r}  \right)  &= \sum_{m=1}^{M} w_m \psi_m \left(  \vec{r}  \right), \\
\psi_m \left(  \vec{r}  \right) & =  \psi \left(  \vec{r}, \vec{\Omega}_m \right).
\end{split}
\label{eq:angular_quad}
\end{align}
Evaluating Equation \ref{eq:mono_transport} at the quadrature points results in M coupled equations that are discrete in direction,
\begin{align}
\vec{\Omega} \cdot \nabla \psi_m \left(  \vec{r} \right) + \sigma_t \psi_m \left(  \vec{r} \right) = 
\frac{\sigma_s}{4\pi}\phi \left(  \vec{r}  \right) + \frac{q}{4\pi}. \label{eq:discrete_in_angle_transport}
\end{align}
In principle, Equation \eqref{eq:discrete_in_angle_transport} could be discretized in space and solved for all
$\{\psi_m\}$; however, the resulting linear system would be have $M/h^3$ degrees of freedom
(DOFs), where $M$ is the total number of directions and $h$ the mesh spacing, which is too large for practical calculations.
A common technique to avoid the large coupled system is to solve these equations iteratively, which is referred to as
``source iteration.'' Source iteration involves lagging the scalar flux, first updating $\{\psi_m\}$ for the $M$ independent
angles and proceeding to update the scalar flux, that ism solve
\begin{align}
\vec{\Omega}_m \cdot \nabla \psi_m^{n+1} \left(  \vec{r} \right) + \sigma_t \psi_m^{n+1} \left(  \vec{r} \right) = 
\frac{\sigma_s}{4\pi}\phi^{n} \left(  \vec{r}  \right) + \frac{q}{4\pi} \label{eq:sn_transport}
\end{align}
for $m=1,...,M$, and form $\phi^{n+1}$ based on $\{\psi_m^{n+1}\}$ from \eqref{eq:angular_quad},
where n is the iteration index.

The upwind discontinuous Galerkin (DG) discretization is investigated exclusively in this paper \cite{PSMIR_1974___S4_A8_0}.
The upwind DG weak formulation for the steady state transport equation for each direction $\Omega$ takes the weak form
\begin{equation}
\label{eq:weak_form_transport}
\begin{aligned}
&-\sum_{T \in T_h }\left( u_h^{n+1},\Omega \cdot \nabla v_h \right) -\sum_{F \in F^i_h}\left\langle ^-u^{n+1}_h,\Omega \cdot \left[ v_h \textbf{n}\right] \right\rangle_F +\sum_{T \in T_h }\left(\sigma_t  u^{n+1}_h, v_h \right)+
\\
&\hspace{20ex} \left\langle u^{n+1}_h,v_h\Omega \cdot \textbf{n} \right\rangle _{\Gamma^+} = -\left\langle g, v_h \Omega \cdot \textbf{n} \right\rangle _{\Gamma ^-} + \sum_{T \in T_h }\left(q^n_{total}, v_h \right).
\end{aligned}
\end{equation}
Here, $n$ and $n+1$ are iteration indices, $(\cdot,\cdot)_T$ are inner products over the volume, and $\langle\cdot,\cdot\rangle_T$
are inner products over the faces. The jump term is given by $\left[ v_h \textbf{n} \right] = v_h \textbf{n}^+ - v_h\textbf{n} ^-$,
where $+$ and $-$ superscripts refer to the upwind and downwind faces respectively, $q_{total}$ is the total source term, including a volumetric source and scattering source, and with the scalar flux lagged as shown in Equation \ref{eq:sn_transport}. Last,
$\Gamma^-$ is the inflow boundary, $\Gamma^+$ is the outflow boundary, and $g$ is a boundary source on the inflow boundary. 

For simplicity, only isoparametric multilinear quadrilateral and hexahedral elements are considered in this paper.
However, the multigrid solver used here has shown to be effective on higher order elements \cite{AIR1,AIR2} and
curvilinear meshes \cite{19sweep} as well, and results presented here are expected to extend to arbitrary mesh and
element order.

\subsection{Solution of the Transport Equation by Sweeping}
\label{sec:sn_by_parallel_sweeps}
Upwind discretizations (such as upwind DG) are typically used to discretize Equation \ref{eq:discrete_in_angle_transport}
for each $m$, because the resulting matrix is logically block lower triangular. A forward solve, or ``sweep'' in transport
literature, can then be used to directly invert the matrix. Although a forward solve is a sequential process, parallel algorithms
have been developed to perform a forward solve for the sparse discretizations resulting from a fixed angular flux direction.
One popular method for parallel sweeping is the KBA method proposed by Koch, Baker, and Alcouffe \cite{KBA}.
For efficiency, this method relies on a specific decomposition of the spatial cells among the processors being used to
solve the problem, as well as an ordering of cells for a given angle such that the matrix is lower triangular in that ordering.
On a structured mesh, the sweep ordering of cells in a uniform mesh is trivial to determine, but on an unstructured mesh
determining a good ordering is a difficult problem (for example, see \cite{Liu_sweeps_unstructured}). 

The parallel sweeping algorithm is not particularly efficient for one direction. However, parallel efficiency of such
algorithms is improved by solving multiple directions at a time. When multiple directions are solved simultaneously, a
scheduling algorithm is generally required to handle the situation where solution fronts for different angular flux directions collide.
This situation was analyzed and a provably optimal scheduling algorithm for uniform meshes developed in \cite{bailey1,adams1}.
The number of communication stages for such algorithms on a uniform mesh is given by $\mathcal{O}(M + dP^{1/d})$,
for $M$ angles, $P$ processors, and a $d$-dimensional problem \cite{bailey1}. For a large number of angles $M$, the
$dP^{1/d}$ term can be masked by $M$. However, for problems with a small to moderate number of angles, or a very
large number of processors, $dP^{1/d}$ is suboptimal parallel scaling. Implementing an efficient parallel sweep on
unstructured meshes is more complicated and remains an area of active research. A review of developments towards
efficient parallel sweeps on unstructured meshes is given in \cite{Liu_sweeps_unstructured},  but it should be noted
that asymptotic complexity is at very best equivalent to that on uniform meshes.

In general, solution of the transport equation requires extensive computing time, up to 90\% of wallclock time
in large multiphysics simulations. As increasing numbers of processors are available for computation, it is desirable to
have an alternative to parallel sweeps with parallel complexity logarithmic in $P$ rather than polynomial in $P$, as well
as black-box sweeping capabilities that are robust and independent of meshing and discretization.

\subsection{Use of Algebraic Multigrid (AMG)} \label{sec:AIR_introduction}

This paper considers the use of an algebraic multigrid (AMG) method in place of traditional parallel sweeping,
motivated by two factors. First, AMG methods have better optimal scalability than parallel sweeping algorithms, and second,
AMG methods can be easily used in a black box manner and do not require specialized domain partitioning.
AMG methods are widely used to solve linear systems arising from the discretization of elliptic and parabolic partial differential equations.
Ideally, AMG scales linearly with the problem size in floating point operations, and communication cost for parallel
AMG scale logarithmically with the number of processors \cite{falgout1}.
In an optimal setting, the time to solution of parallel AMG scales like $\log(P)$, for $P$ processors.
In practice, factors such as growth of the convergence factor, growth of problem size, and coarse-grid fill in can lead to scaling $\log^m(P)$ for $m\geq 1$.
Nevertheless, for reasonable $m$, AMG scalability is still asymptotically better than the scaling of transport sweeps, $dP^{1/d}$.
Recently, a classical AMG method based on an approximate ideal restriction (AIR) was developed for nonsymmetric matrices.
AIR has shown to be effective for solving linear systems arising from upwind discontinuous Galerkin (DG) finite element
discretizations of advection-diffusion problems, including the hyperbolic limit of pure advection \cite{AIR1,AIR2}.
A parallel version of AIR (pAIR) is now available in the \textit{hypre} library, and this paper investigates the performance of
pAIR for solving the S$_N$ transport equations.

The AIR algorithm and supporting theory are developed in \cite{AIR1} and \cite{AIR2}. Briefly,
AIR is a Petrov-Galerkin AMG method based on the construction of an approximation to a certain ideal restriction operator.
The ideal restriction operator exactly removes error modes from the coarse grid that are not in the range of interpolation, or,
equivalently, provides an exact correction at coarse-grid points. 
Such a restriction operator separates the coarse-grid problem from the fine-grid problem, so solving $A\mathbf{x} = \mathbf{b}$
is reduced to solving two smaller problems. AIR approximates the ideal restriction operator and is thus an approximate
reduction-based AMG method. AIR is coupled with a simple interpolation, where coarse-grid points are interpolated by
value to the fine grid, and several relaxations over F-points on the fine grid then complement the correction to C-points
provided by AIR. Two methods were developed in \cite{AIR1,AIR2} for computing $R$, the local approximate
ideal restriction ($\ell$AIR) and a method based on the use of a finite Neumann expansion ($n$AIR), both of which
are examined in this paper. For more details on AIR, please see \cite{AIR1,AIR2}.

\subsection{Sweeping with pAIR}
\label{sec:sweeping_with_pAIR}

This paper studies the performance of pAIR in S$_N$ transport simulations, replacing a traditional sweep in source
iteration with a pAIR solve for each angle. Section \ref{sec:perf} examines the performance of pAIR on different
problems and meshes and compares different parallel relaxation routines. One unique aspect of pAIR over a
traditional sweep is that the user has flexibility in how accurately the system is solved. In Section \ref{sec:accuracy},
a short algebraic analysis followed by numerical tests indicate that only a few pAIR iterations are necessary to
reach discretization accuracy. Weak scaling is then discussed in Section \ref{sec:weak_scaling_parallel_in_space},
and angular parallelism introduced in Section \ref{sec:multiple_meshes}, where multiple MPI groups store an entire
 copy of the spatial mesh, and each group solves a subset of the angles,
which reduces the total time to solution by 30--50\%. Finally, weak scaling and a comparison with traditional
sweeps is presented in Section \ref{sec:comp_with_sweeps}

pAIR has been implemented in the publicly available \textit{hypre} library (\url{https://github.com/hypre-space/hypre})
\cite{Falgout:2002vu}, which is used exclusively in this paper. 
The program created to solve the $S_N$ equations relies on the deal.II finite element library \cite{dealII90}.
For domain decomposition and manipulation of a fully distributed mesh, deal.II relies on the p4est library \cite{BursteddeWilcoxGhattas11} and
for fully distributed matrices and vectors, deal.II is used with the Trilinos library \cite{Trilinos_general}. The Trilinos ifpack package \cite{ifpack-guide} is used to interface with \textit{hypre}.
Note that the linear systems investigated in this paper have a natural block-diagonal structure.
For the two-dimensional quadrilateral and three-dimensional hexahedral multi-linear elements,
block sizes are four and eight respectively. In principle the block structure can be accounted for directly by 
the pAIR algorithm (as discussed in \cite{AIR1} and \cite{AIR2}). However, the parallel matrix structures used
in the software do not contain block information, so matrices are instead scaled by their block inverse, resulting
in a logically lower triangular matrix with a unit diagonal.

\section{pAIR Performance for Transport}\label{sec:perf}

AIR has proven an effective solver for steady state advection through a medium with discontinuities in the
total cross section, using both a uniform mesh and an unstructured mesh \cite{AIR2}. 
In the works developing the AIR algorithm and supporting theory, simple Jacobi relaxation was shown in practice
to be effective and some theoretical basis for the effectiveness was provided.
For all the problems investigated as part of this paper, Jacobi relaxation is again found to be a robust relaxation scheme.
However, the logically lower triangular nature of transport discretizations investigated in this
paper allows for potentially more effective relaxation schemes to be implemented, specifically, an on-processor
forward solve can be used as a relaxation method, which exactly inverts the principle submatrix stored on
processor, 

In Section \ref{sec:perf:mesh}, the performance of pAIR is compared and demonstrated on a variety of meshes
and problem types for the two types of relaxation. Then, pAIR is  compared with only using on-processor
relaxation as a preconditioner for GMRES (without AMG) in Section \ref{sec:perf:jacobi}.
In practice a transport sweep is sometimes replaced by an on-processor solve either due to ease of implementation
or in an attempt to reduce the total time to solution, particularly on unstructured meshes, where sweeps can be
very time consuming. Here, we look at the regimes of total opacity where pAIR is better than an on-processor solve.

\subsection{Assessment of pAIR with Different Meshes and Cross-Section Characteristics}\label{sec:perf:mesh}

In two dimensions, pAIR is tested on four different meshes, a uniform structured mesh and the three
meshes shown in Figure \ref{fig:mesh_schematics}. The unstructured mesh was generated by meshing the
surface structure shown in Figure \ref{fig:domain_unstruct_mesh} using the Gmsh mesh generator program.
The zmesh shown in Figure \ref{fig:2D_zmesh} is logically a structured mesh and meant primarily to test high
aspect ratio elements \cite{kershaw1}.
The meshes are distributed over 540 processors and refined to approximately $70,000$ spatial DOFs per
processor. A sweep is then performed by solving for 32 directions in an angular quadrature set sequentially.
This involves constructing a global FEM matrix for a single direction, solving the linear system, then clearing the matrix and repeating these steps for the next direction.
For the tests presented in this section, the linear system is solved to machine precision.

\begin{figure}[!htbp]
	\centering
	\begin{subfigure}[b]{0.25\textwidth}
		\centering
		\includegraphics[width=\textwidth]{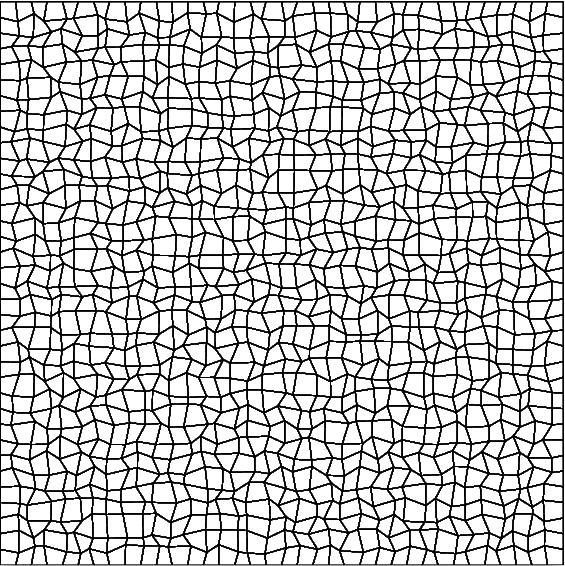}
		\caption{}
		\label{fig:2D_random_vert}
	\end{subfigure}
	\hspace{2ex}
	\begin{subfigure}[b]{0.25\textwidth}
		\centering
		\includegraphics[width=\textwidth]{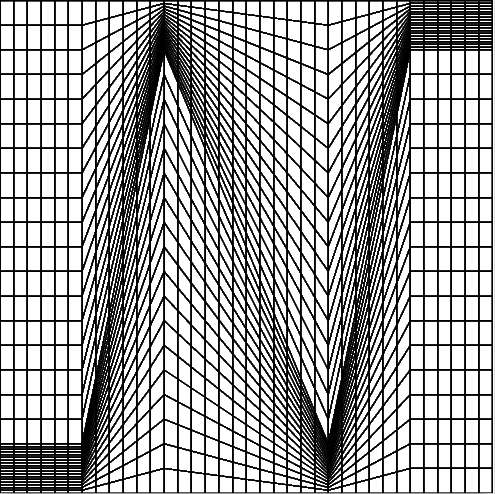}
		\caption{}
		\label{fig:2D_zmesh}
	\end{subfigure}
	\hspace{2ex}
		\begin{subfigure}[b]{0.25\textwidth}
		\centering
		\includegraphics[width=\textwidth]{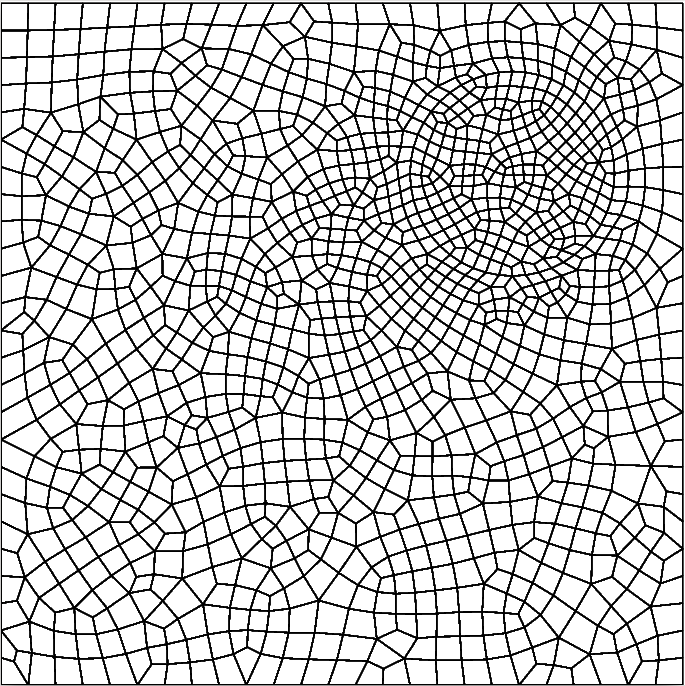}
		\caption{}
		\label{fig:2D_randomized_gmsh}
	\end{subfigure}
	\\
	\begin{subfigure}[b]{0.25\textwidth}
		\centering
		\includegraphics[width=\textwidth]{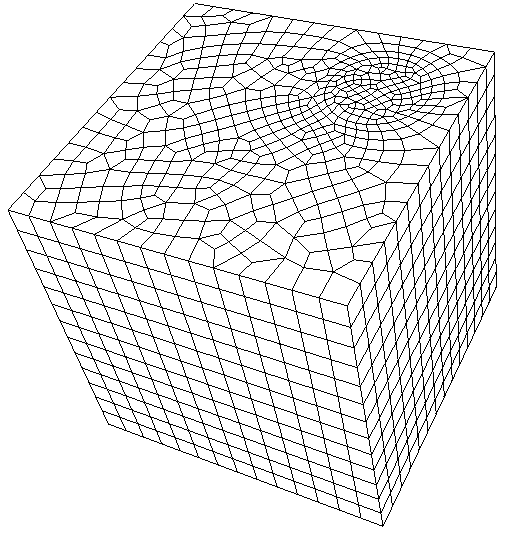}
		\caption{}
		\label{fig:3D_randomized_gmsh}
	\end{subfigure}
	\hspace{2ex}
	\begin{subfigure}[b]{0.25\textwidth}
	\centering
	\includegraphics[width=\textwidth]{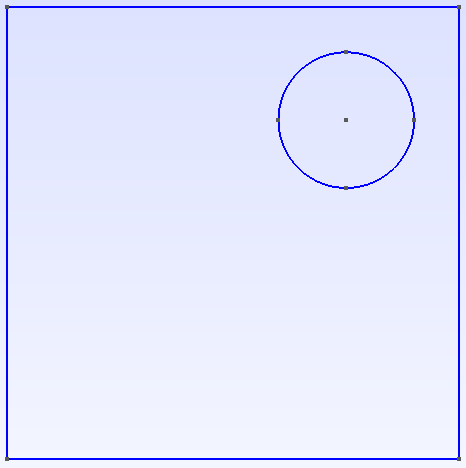}
	\caption{}
	\label{fig:domain_unstruct_mesh}
	\end{subfigure}
	\caption{Four meshes on which pAIR was tested: \protect\subref{fig:2D_random_vert} logically uniform two-dimensional mesh with randomized vertices; \protect\subref{fig:2D_zmesh} two-dimensional zmesh \cite{kershaw1}; \protect\subref{fig:2D_randomized_gmsh} unstructured two-dimensional mesh; \protect\subref{fig:3D_randomized_gmsh} unstructured extruded three-dimensional mesh. }
	\label{fig:mesh_schematics}
\end{figure}

Table \ref{table:2Ddomains_tested} provides the different material compositions tested, with the box and banded
configurations depicted in Figure \ref{fig:domain_schematics}.
The square domain has a side length 100 cm and a volumetric source of 1 $\frac{n}{cm^2}$.
Inhomogeneous cross sections are implemented such that the cross section is constant within each cell.

\begin{figure}[!h]
	\begin{center}
	\includegraphics[scale=0.35]{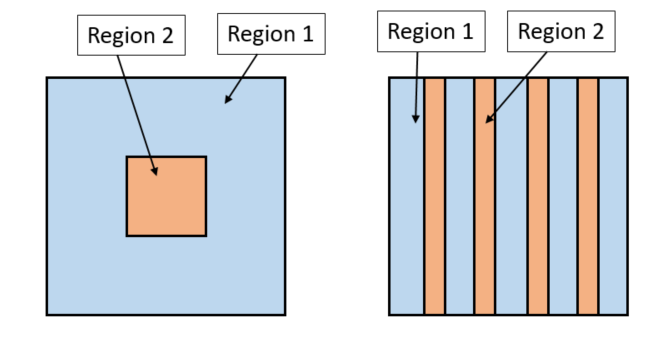}
	\caption{Schematic of the Box1, Box2, and Banded configuration}
	\label{fig:domain_schematics}

	\captionof{table}{Cross-section Values for the Different Material Configurations Tested}\label{table:2Ddomains_tested}
	\begin{tabular}{|c|c|}
	\hline
		\textbf{Material Type} & \textbf{Total Cross Section}  \\\hline
		\hline
		
		Box1 & Inner Box 0.01  $cm^{-1}$, Outer Box 100.0  $cm^{-1}$ \\ \hline
		Box2 & Inner Box 100.0 $cm^{-1}$, Outer Box 0.01 $cm^{-1}$ \\ \hline
		Banded & Region 1 0.01 $cm^{-1}$, Region 2 100.0 $cm^{-1}$ \\ \hline
	\end{tabular}
	\end{center}
\end{figure}

Average convergence factors for the homogeneous domain problems are presented in Table
\ref{tab:2D_varried_homogenous_average_convergence} and the average convergence factors for the problems with
varied cross sections are presented in Table \ref{tab:2d_varried_box_banded_average_convergence_factors}.
For the homogeneous problems, the total time for the sweep update to complete are shown in Table \ref{tab:2d_varried_homogenous_total_solve_time}.
Average convergence factors are computed by averaging the convergence factors for each single direction solve.

\begin{table}[!h]	\centering
	\caption{Average Convergence Factors for a Full Sweep on the Two-Dimensional Meshes with Homogeneous Cross-Sections}
	\label{tab:2D_varried_homogenous_average_convergence}
	\begin{tabular}{|c|c|c|c|c|c|}
		\hline
		Mesh Type   & Relaxation & $\sigma_t=0.0$ & $\sigma_t=0.01$ & $\sigma_t=1.0$ & $\sigma_t=100.0$ \\ \hline \hline
		
		\multirow{2}{*}{Random Vertices}   & Jacobi & 0.080 & 0.080 & 0.071 & 0.018 \\ \cline{2-6}
		&  On Proc Solve & 0.043 & 0.044 & 0.037 & 0.0013 \\ \hline
		
		\multirow{2}{*}{Uniform}  & Jacobi & 0.036 & 0.036 & 0.030 & 0.0023 \\ \cline{2-6}
		&  On Proc Solve & 0.018 & 0.018 & 0.018 & 1.1e-4 \\ \hline
		
		\multirow{2}{*}{zmsh}  & Jacobi  & 0.29~ & 0.29~ & 0.28~ & 0.17~ \\ \cline{2-6}
		&  On Proc Solve & 0.052 & 0.053 & 0.048 & 0.010 \\ \hline
		
		\multirow{2}{*}{gmsh}  & Jacobi  & 0.22~ & 0.22~ & 0.21~ & 0.13~ \\ \cline{2-6}
		&  On Proc Solve & 0.033 & 0.033 & 0.028 & 9.4e-4 \\ 
		
		 \hline
	\end{tabular}
\end{table}
\begin{table}[h!]	\centering
	\caption{Total Solve Time for a Full Sweep on the Two-Dimensional Meshes with Homogeneous Cross-Sections}
	\label{tab:2d_varried_homogenous_total_solve_time}
	\begin{tabular}{|c|c|c|c|c|c|}
		\hline
		Mesh Type   & Relaxation & $\sigma_t=0.0$ & $\sigma_t=0.01$ & $\sigma_t=1.0$ & $\sigma_t=100.0$ \\ \hline \hline
		
		\multirow{2}{*}{Random Vertices}   & Jacobi & 20.5 s& 20.2 s& 19.1 s& 10.8 s\\ \cline{2-6}
		&  On Proc Solve & 17.4 s& 17.2 s& 16.9 s& 8.63 s\\ \hline
		
		\multirow{2}{*}{Uniform}  & Jacobi & 11.2 s& 11.1 s& 11.0 s& 5.82 s\\ \cline{2-6}
		&  On Proc Solve & 11.4 s& 11.3 s& 10.1 s& 4.54 s \\ \hline
		
		\multirow{2}{*}{zmsh}  & Jacobi  &34.4 s& 34.3 s& 33.4 s& 21.5 s\\ \cline{2-6}
		&  On Proc Solve & 16.5 s& 16.4 s& 15.9 s& 10.0 s\\ \hline
		
		\multirow{2}{*}{gmsh}  & Jacobi  & 15.7 s& 15.5 s& 15.7 s& 9.70 s\\ \cline{2-6}
		&  On Proc Solve & 8.27 s& 8.25 s& 8.16 s& 4.02 s\\ 
		
		\hline
	\end{tabular}
\end{table}
\begin{table}[h!]	\centering
	\caption{Average Convergence Factors for a Full Sweep on the Two-Dimensional Meshes with Inhomogeneous Cross-Sections}
	\label{tab:2d_varried_box_banded_average_convergence_factors}
	\begin{tabular}{|c|c|c|c|c|}
		\hline
		Mesh Type   & Relaxation & box1 & box2 & banded  \\ \hline \hline
		
		\multirow{2}{*}{Random Vertices}   & Jacobi & 0.079 & 0.041 & 0.043 \\ \cline{2-5}
		&  On Proc Solve & 0.042 & 0.0075 & 0.0064 \\ \hline
		
		\multirow{2}{*}{Uniform}  & Jacobi & 0.034 & 0.010 & 0.0098 \\ \cline{2-5}
	
		&  On Proc Solve & 0.019 & 0.0012 & 9.7e-04 \\ \hline
		
		\multirow{2}{*}{zmsh}  & Jacobi  & 0.29 & 0.19 & 0.27 \\ \cline{2-5}
		&  On Proc Solve & 0.048 & 0.010 & 0.032 \\ \hline
		
		\multirow{2}{*}{gmsh}  & Jacobi  & 0.021 & 0.015 & 0.018 \\ \cline{2-5}
		&  On Proc Solve & 0.0032 & 0.0046 & 0.0066\\
		
		\hline
	\end{tabular}
\end{table}

For all cases tested, simple Jacobi relaxation results in reasonable convergence factors, but the on-processor
solve as a relaxation always outperforms Jacobi relaxation.
The convergence factors for the unstructured mesh are similar to those of the two logically rectangular grids.
For all mesh types investigated, results for the Box1 configuration are very similar to the homogeneous domain
test with a cross section of 100.0 $cm^{-1}$.
The Box1 configuration has a large outer opaque region and a smaller inner transparent region.
The Box2 configuration results are more similar to the homogeneous results with a cross section of 0.01 $cm^{-1}$.
An important factor for the solver appears to be the number of optically thin cells, suggesting that results from
homogeneous (thin) domains should be a worst-case on heterogeneous domains, consistent with results in \cite{AIR1}.

In three dimensions, a uniform structured grid and an unstructured extruded mesh are investigated, with approximately
16,600 spatial DOFs per processor. The unstructured grid is shown in Figure \ref{fig:mesh_schematics}.
Average convergence factors are presented in Table \ref{tab:2d_varried_box_banded_average_convergence_factors} and
Table \ref{tab:3D_varried_homogenous_average_convergence}.
Trends similar to those seen in the two-dimensional results are again apparent in the three dimensional results.

\begin{table}[h!]	\centering
	\caption{Average Convergence Factors for a Full Sweep on the Three-Dimensional Meshes}
	\label{tab:3D_varried_homogenous_average_convergence}
	\begin{tabular}{|c|c|c|c|c|c|}
		\hline
		Mesh Type   & Relaxation & $\sigma_t=0.0$ & $\sigma_t=0.01$ & $\sigma_t=1.0$ & $\sigma_t=100.0$ \\ \hline \hline
		
		\multirow{2}{*}{Uniform}   & Jacobi & 0.14 & 0.14 & 0.13 & 0.034 \\ \cline{2-6}
		&  On Proc Solve & 0.075 & 0.075 & 0.069 & 0.0053 \\ \hline
		
		\multirow{2}{*}{Unstructured}  & Jacobi & 0.17 & 0.17 & 0.16 & 0.043 \\ \cline{2-6}
		&  On Proc Solve & 0.010 & 0.010 & 0.094 & 0.0081 \\
		
		\hline
	\end{tabular}
\end{table}

\subsection{Comparison of pAIR with Domain-Wise Block Jacobi}\label{sec:perf:jacobi}

The previous section compared pAIR time to solution when using Jacobi relaxation and an on-processor
solve for relaxation. In some cases, the transport sweep in source iteration is actually replaced with the same
on-processor solve as an approximate inverse (also known as inexact parallel block Jacobi)
\cite{warsa2003improving}. This can be thought of as a block Jacobi iteration,
with blocks given by the domain each processor owns. Here, we make a direct comparison between pAIR using an
on-processor solve for relaxation, and GMRES preconditioned with an on-processor solve.

\begin{figure}[!hbt]
	\centering
	\begin{subfigure}[b]{0.475\textwidth}
		\centering
		\includegraphics[width=\textwidth]{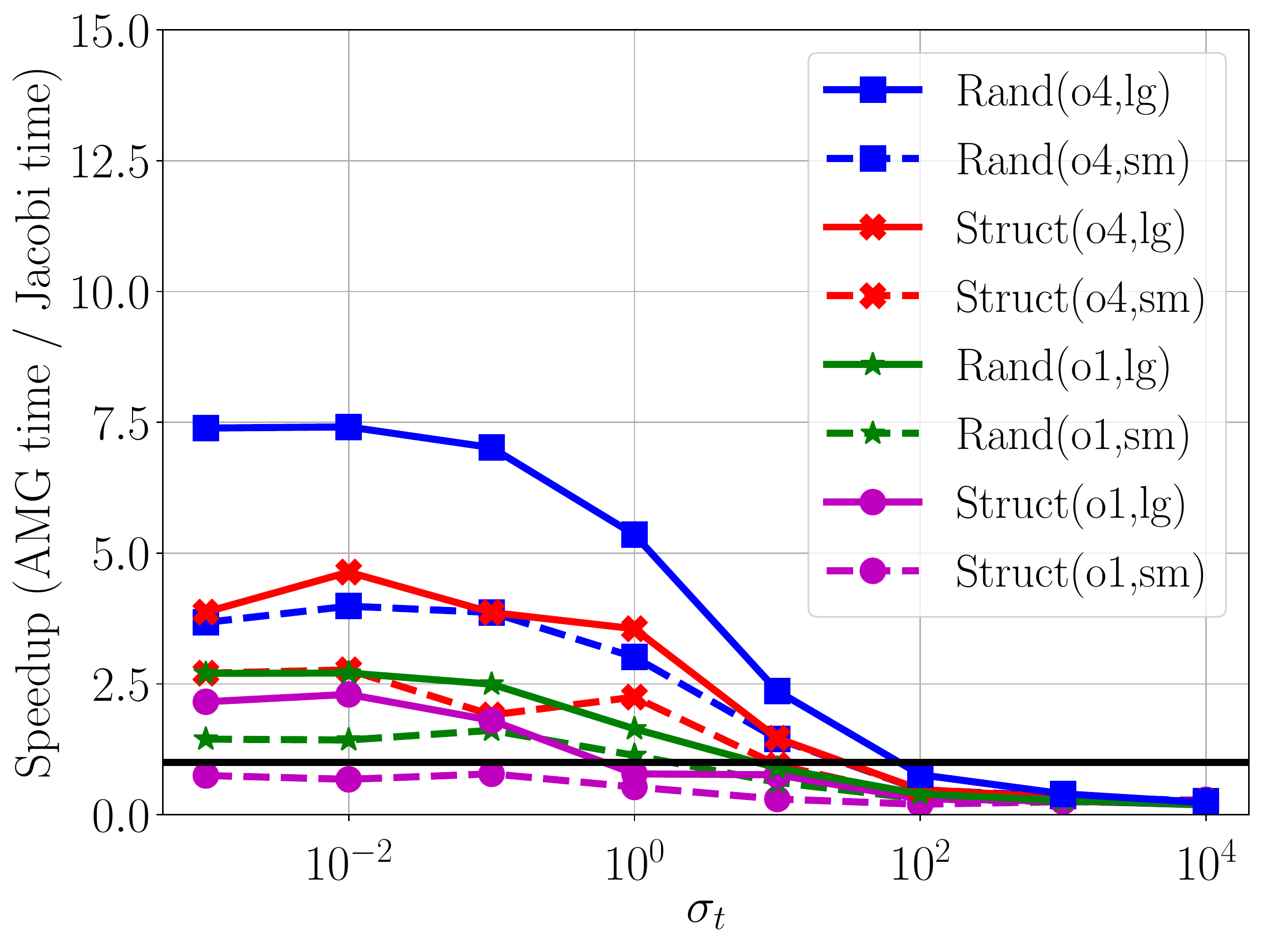}
		\caption{512 Processors}
		\label{fig:comp_jacobi_512}
	\end{subfigure}
	\begin{subfigure}[b]{0.475\textwidth}
		\centering
		\includegraphics[width=\textwidth]{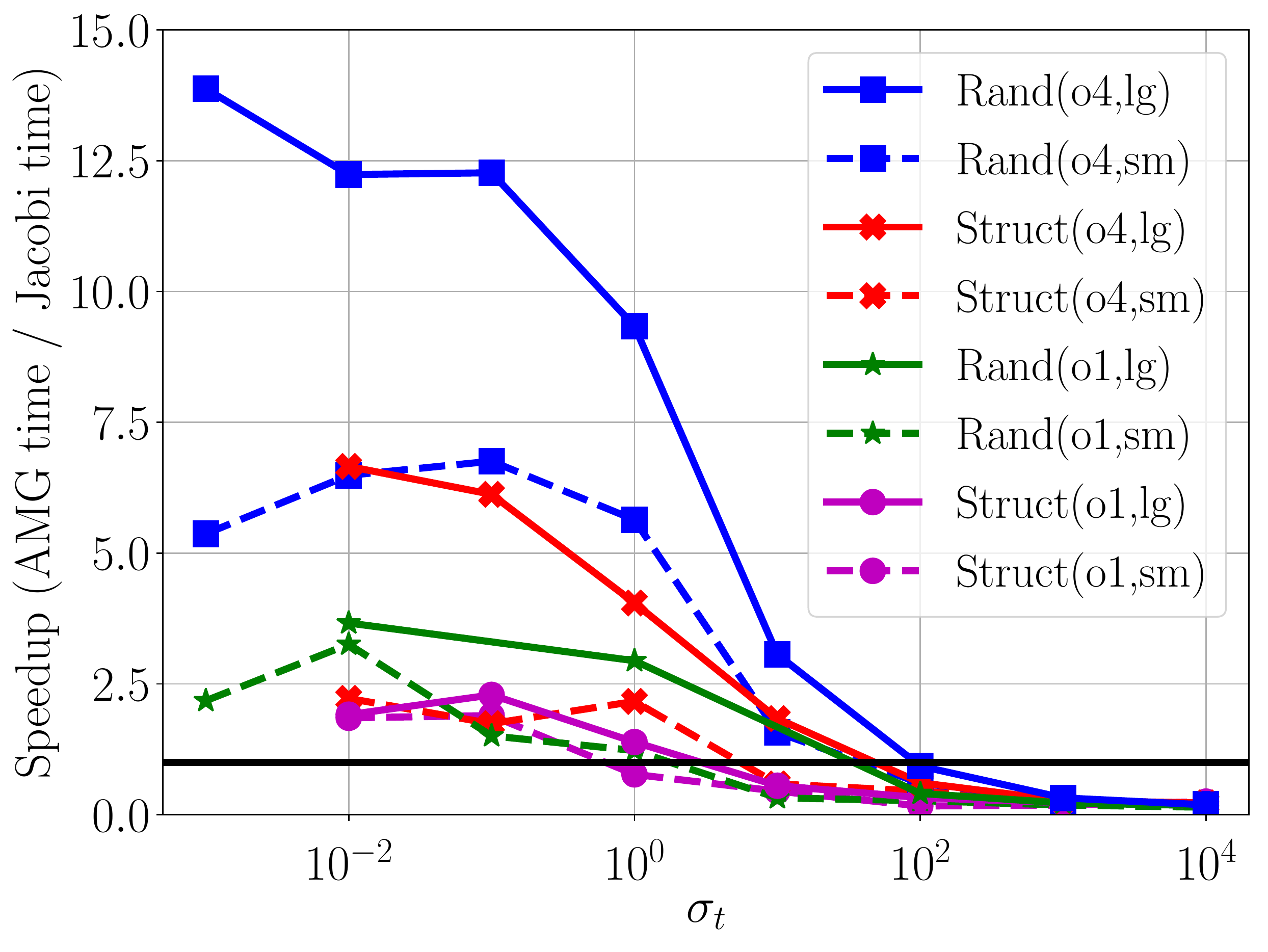}
		\caption{2048 processors}
		\label{fig:comp_jacobi_2048}
	\end{subfigure}
	  \caption{Comparison of pAIR and GMRES preconditioned with an on-processor solve, as a function of total opacity
  $\sigma_t$, for first- (o1) and fourth-order (o4) finite elements, on structured (struct) and unstructured (rand)
  meshes, and for a small (sm) and large (lg) on-processor problem size. Local DOFs range from 8,000 for Struct(o1,sm)
  to 700,000 for Rand(o4,lg).}
\label{fig:comparison}
\end{figure}

Results are based on a simple 2d test problem with homogeneous total cross section
$\sigma_t = C$, for some constant $C$, and direction $\Omega = (\cos(\theta),\sin(\theta))$, with
$\theta = 3\pi/16$. Figure \ref{fig:comparison} shows the speedup that pAIR provides over GMRES preconditioned
with on-processor solve, as a function of $\sigma_t$, and for a multiple mesh types, finite-element orders, and local
problem sizes. For $\sigma_t \gg 1$, so-called ``optically thick
regimes,'' the matrix is already well-conditioned and an on-processor relaxation converges rapidly without
multigrid. For thin regimes, $\sigma_t < 1$, however, pAIR leads to as much as a $14\times$ speedup
over GMRES and the on-processor solve. These tests only went up to 2048 processors, but as can be
seen, this speedup will continue to grow as more and more processors are used in a weak scaling sense
(because block-processor is a Jacobi iteration, and convergence deteriorates as problem size increases),

Note, when an on-processor solve is used in source iteration, typically only one or a few iterations are performed
for each source iteration (that is, it is not applied until convergence). However, if there are thin regions in the
domain, convergence of the larger source iteration will likely be more-or-less defined by convergence of the
block Jacobi algorithm on the thin regions (where convergence will be slowest). The following section
investigates how accurately systems must be solve by pAIR.

\section{Accuracy of Spatial Solves} \label{sec:accuracy}

A unique aspect of using an iterative method to solve the linear transport equation compared with traditional sweeps
(which are an exact solve) is the question as to how accurately each linear system should be solved. In fact, this question
has subtle practical implications, and how accurately the systems are solved actually defines the solution that source
iteration converges to.

Consider $M=3$ angles and one energy group. Then the full discretized set of equations, replacing the scattering integral with
a sum over angular quadrature weights $\{\omega_i\}$ and corresponding spatial discretization $\{\mathcal{L}_i\}$, 
can be written as a block linear system,
\begin{align*} 
\begin{bmatrix} \mathcal{L}_1 &&& - \sigma_s I \\
&\mathcal{L}_2&& -\sigma_s I\\
&&\mathcal{L}_3& -\sigma_s I \\
\omega_1 I &  \omega_2 I &  \omega_3 I &  -I \end{bmatrix}
\begin{bmatrix} \psi_1 \\ \psi_2 \\ \psi_3 \\ \varphi \end{bmatrix} = \begin{bmatrix}q_1 \\ q_2 \\ q_3 \\ 0\end{bmatrix},
\end{align*}
where $\mathcal{L}_i \sim \underline{\Omega_i}\cdot\nabla + \sigma_t$. For memory purposes, the standard approach
in transport simulation is to eliminate the block diagonal matrix corresponding to the angular flux vectors, and iterate
only on the scalar flux $\varphi$. This corresponds to a Schur complement problem for the scalar flux,
$\mathcal{S}\varphi = \mathbf{b}$, where
\begin{align*}
\mathcal{S} & := I - \begin{bmatrix} \omega_1 I &  \omega_2 I &  \omega_3 I I \end{bmatrix}
	\begin{bmatrix}  \mathcal{L}_1^{-1} & \\ &\mathcal{L}_2^{-1}&\\ &&\mathcal{L}_3^{-1} \end{bmatrix} \begin{bmatrix}
	\sigma_s I \\  \sigma_s I \\  \sigma_s I\end{bmatrix} 
= I - \sum_{i=1}^f \omega_i\sigma_s\mathcal{L}_i^{-1}.
\end{align*}
Source iteration corresponds to performing Richardson iteration on $\mathcal{S}\varphi = \mathbf{b}$,
\begin{align*}
\varphi^{(i+1)} & = \varphi^{(i)} + \mathbf{b} - \mathcal{S}\varphi^{(i)} 
	 = \mathbf{b} + \sum_{i=1}^f \omega_i\sigma_s\mathcal{L}_i^{-1}\varphi^{(i)}.
\end{align*}
For a detailed discussion on transport iterations in a linear algebra framework, see \cite{Faber:1989wo,DSA19}.

Note that a transport sweep, that is computing the action of $\mathcal{L}_i^{-1}$ for $i=1,...,M$, is used
to compute \textit{the action} of the operator $\mathcal{S}$. Now suppose that we do not invert $\mathcal{L}_i$
exactly and instead apply $\widehat{\mathcal{L}}_i^{-1}$, denoting some $\ell$ iterations of pAIR for $\mathcal{L}_i$.
This corresponds to applying Richardson iteration to the modified problem
$\widehat{\mathcal{S}}\varphi = \mathbf{b}$, where $\widehat{S}:= I - \sum_{i=1}^3 \omega_i\sigma_s\widehat{\mathcal{L}}_i^{-1}$.
Thus in applying approximate inverses $\widehat{\mathcal{L}}_i^{-1}$, corresponding to, say, ten pAIR V-cycles,
we actually converging Richardson/source iteration to a modified angular flux problem, defined by the approximate
inverses.

There is a standard question in finite element discretizations and linear solvers as to how accurately the
liner system must be solved to achieve discretization accuracy (that is, the solution to the linear system is
sufficiently accurate that solving to increased accuracy does not lead to better approximation of the continuous
solution). Although the question is slightly more complicated here due to the approximate $\widehat{\mathcal{S}}$
discussed above, a similar question arises -- how many V-cycles must we apply such that the solution of
$\widehat{\mathcal{S}}\varphi = \mathbf{b}$ is just as good of an approximation to the continuous PDE as
the solution of ${\mathcal{S}}\varphi = \mathbf{b}$?
Some insight into this question is provided in Figure \ref{fig:mms_error_versus_pAIR_iterations_plot}, which
shows the error (from continuous solution) in the computed scalar flux with a fixed number of pAIR iterations, along
with the average final pAIR residual.
The error is computed based on a method of manufactured solutions (MMS) for a three dimensional cube with side length of
1~cm discretized and a structured grid. The MMS source is provided in the Appendix. Also shown in the figure is the relative residual for the final source iteration.

\begin{figure}[!hbt]
	\centering
	\includegraphics[width=0.75\textwidth]{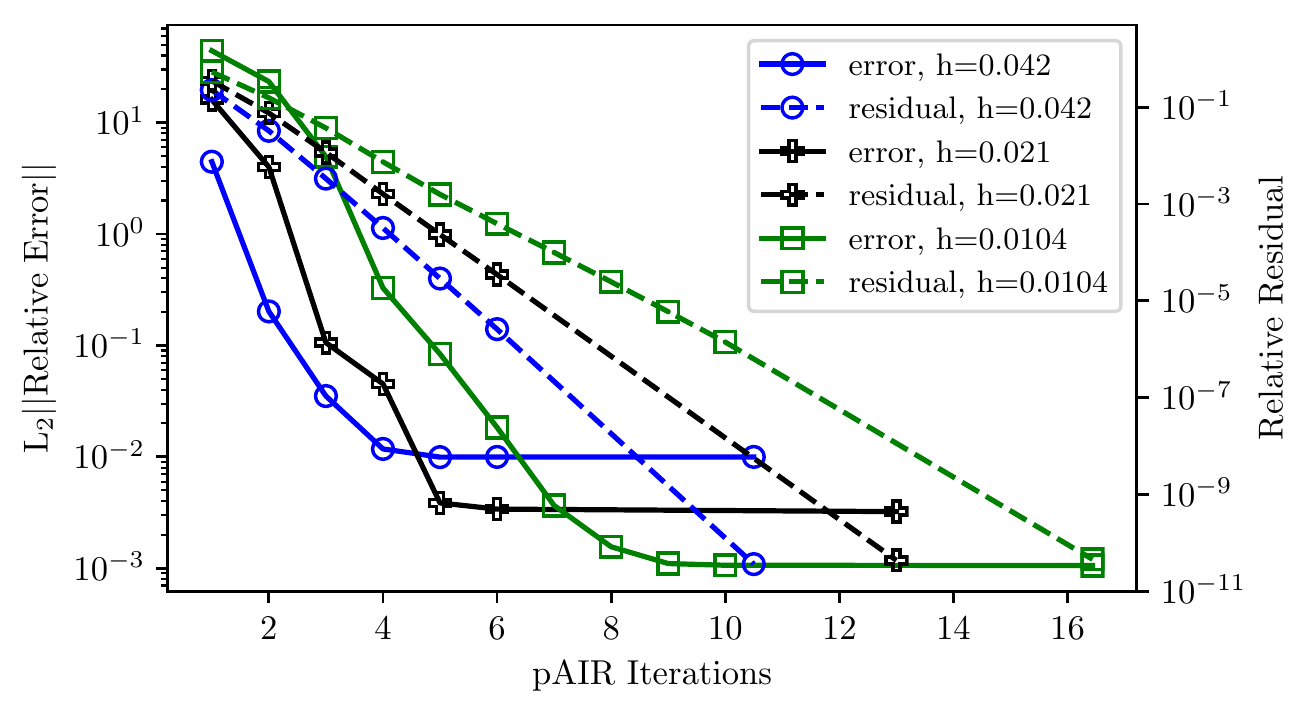}
	\caption{Comparison of relative error in solution and relative residual versus the number of pAIR iterations. Results from three successively refined meshes are shown where h is the spacing between vertices.}
	\label{fig:mms_error_versus_pAIR_iterations_plot}
\end{figure}

As expected, as mesh spacing $h\to 0$, the linear systems must be solved to increasingly accurate
tolerances, because the solution is also increasingly accurate. However, it is clear that the systems do
not need to be solved to machine precision, and only a handful of iterations are necessary for good
convergence.

\section{Weak Scaling with Parallelism in Space}
\label{sec:weak_scaling_parallel_in_space}

As discussed in Section \ref{sec:AIR_introduction}, for an ideal multigrid method, the total time to solution in a weak scaling test
should grow linearly in log-log space, with a slope equal to one when plotted with the logarithm of total problem size (or equivalently
the total number of processors). In practice, factors such as growth of the convergence factor and coarse-grid fill in can lead to
a slope $m$, corresponding to $log^m(P)$, that is greater than one.

Weak scaling results to 4,096 processors for both a three-dimensional uniform grid and unstructured grid are shown in Figure
\ref{fig:3d_scaling_no_filter}.
A uniform volumetric source of strength 1 $\frac{n}{cm^3s}$ was used for both tests.
The calculations are performed with a square Chebyshev-Legendre $S_4$ quadrature set having 32 directions \cite{osti_5958402}.
For the uniform grid, the spatial domain size was fixed at 32,768 DOFs per processor, while the unstructured grid has 50,960 DOFs
per processor. Both of these results were generated with distance-2 $\ell$AIR, pointwise Jacobi relaxation, Falgout coarsening with
a strength of connection (SOC) parameter of 0.25, and a strength parameter for $R$ of $\epsilon_r=0.01$ (see Section
\ref{sec:scaling:soc}). 
In both cases, time to solution grows linear in log-log space for both solve time and setup time,
an important initial demonstration of the method's scalability.

\begin{figure}[!htb]
	\centering
	\includegraphics{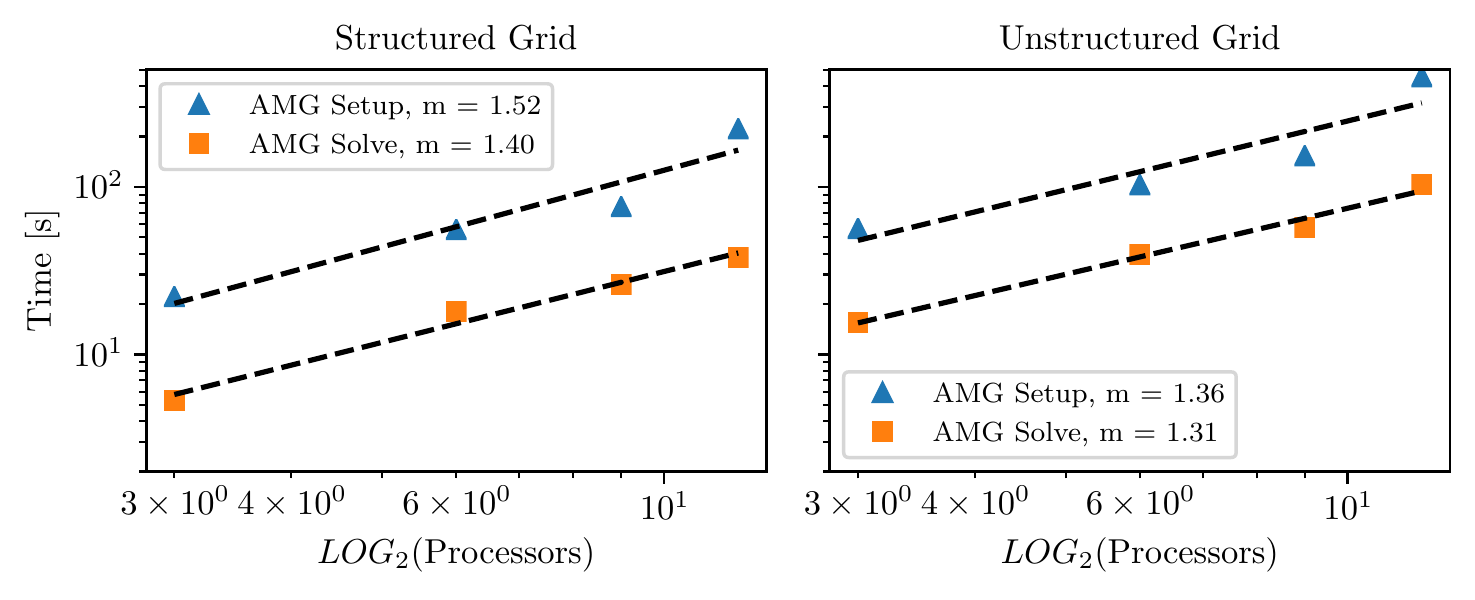}
	\caption{Scaling results for both a uniform grid and unstructured mesh for distance-2 $\ell$AIR with no operator filtering. Time reported is that to solve all 32 directions in the angular quadrature set.}   
	\label{fig:3d_scaling_no_filter}
\end{figure}

\subsection{Relevance of Operator Filtering and SOC Parameter}

The AMG algorithm includes a number of parameters which affect various aspects of performance.
Additionally, pAIR, like other AMG methods, depends on distinct algorithms that accomplish tasks in the overall method, and different algorithms can be used to accomplish the same task.
A parameter study is not included as part of this paper, however, some general considerations regarding two important parameters are discussed in the next two sections.
These two parameters are particularly relevant to scalability and overall time to solution.
They also most markedly affect the total cost of pAIR and should thus be first considered when using pAIR in practice to solve the transport equation.

\subsubsection{SOC Parameter for Building $R$}\label{sec:scaling:soc}
When building the approximate ideal restriction operator $R$, each row corresponds to a C-point and has a nonzero
sparsity pattern of F-point neighbors. Instead of finding the neighborhood of F-points in the matrix $A$, the
neighborhood is based on a SOC matrix, where only ``strongly'' connected F-points are chosen for the sparsity pattern. 
This increases the sparsity of $R$ as well as decreasing the cost of computing $R$. A coefficient is included in the SOC
matrix if it is strongly connected in a classical AMG sense, that is, all columns $\{j\}$ in row $i$ such that

\begin{equation}\label{eq:soc}
|a_{ij}|\geq \epsilon_{R} max_{i\ne j}\left(|a_{ij}|\right.
\end{equation}
Here a new SOC parameter is introduced, $\epsilon_{R}$, that is generally different than the parameter used for coarsening.
Note also that the absolute value of coefficients are used instead of the negative values.

Construction of $R$ using distance-1 and distance-2 $\ell$AIR is discussed in Section \ref{sec:AIR_introduction}.
In general, convergence factors will improve as the SOC parameter $\epsilon_{R}$ is made smaller.
However, the time it takes to build $R$ and apply pAIR will increase as more connections are considered.
Figure \ref{fig:soc_r_study_convergence_factors} shows average convergence factors for different $\epsilon_{R}$ parameter values
for distance-2 and distance-1 $\ell$AIR respectively, and Figure \ref{fig:soc_r_study_setup_versus_total_time} shows
the total pAIR time and setup time for different $\epsilon_{R}$ parameter values for distance-2 $\ell$AIR.
Results are presented for both a uniform grid and the unstructured mesh. Although results shown are for a single direction,
similar results are seen for each direction in the angular quadrature set investigated. 

For both distance-1 and distance-2 $\ell$AIR, as $\epsilon_{R}$ is decreased, the convergence factors decrease.
Additionally, there is less convergence factor growth apparent in the weak scaling test for small values of $\epsilon_{R}$.
Note that for distance-1 $\ell$AIR, the convergence factors stop significantly decreasing as $\epsilon_{R}$ becomes smaller.
This is because after a certain point, most of the distance-1 neighbors are already being considered and decreasing $\epsilon_{R}$
does not include any new neighbors. For distance-2 $\ell$AIR, there are many more possible neighbors and significant
differences can be seen in the average convergence factors even between the two relatively small $\epsilon_{R}$ values of 0.05 and 0.01.
\begin{figure}[!htb]
	\centering
	\includegraphics[width=\textwidth]{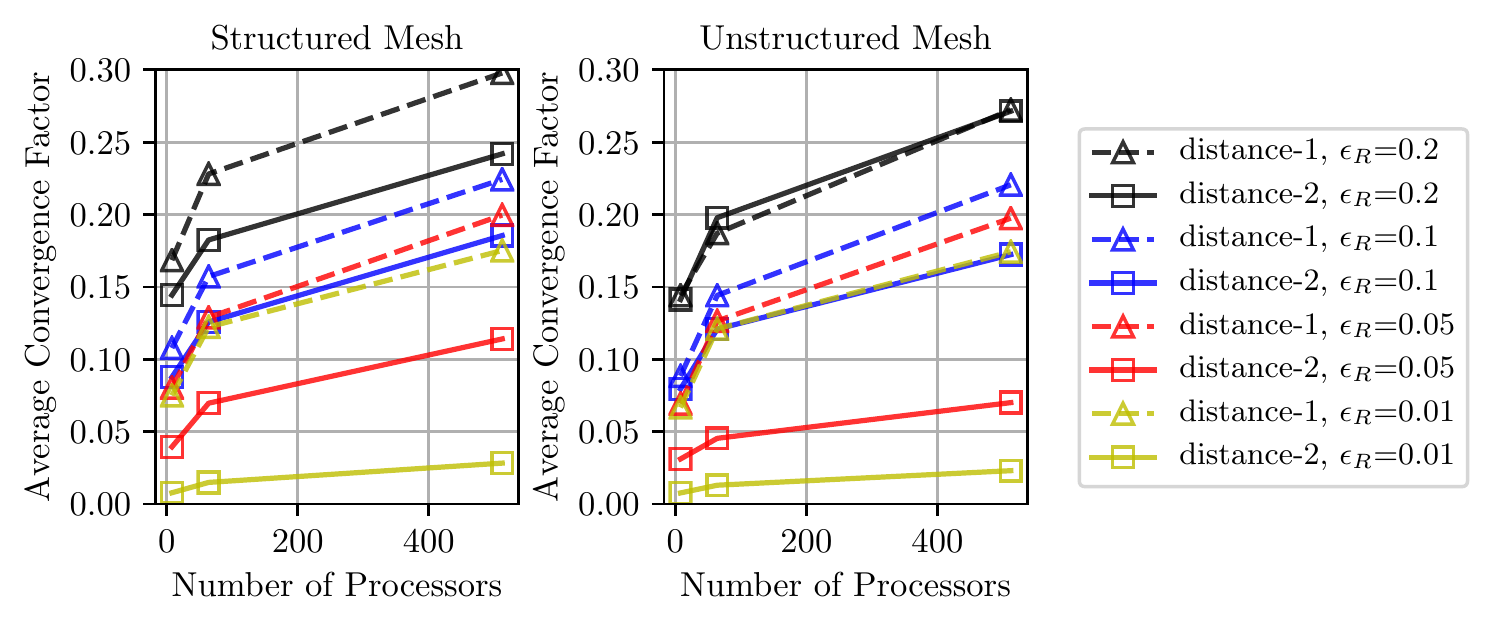}
	\caption{Average convergence factors for both a uniform mesh and unstructured mesh as $\epsilon_{R}$ is made smaller shown for distance-1 and distance-2 $\ell$AIR.}
	\label{fig:soc_r_study_convergence_factors}
\end{figure}

Interestingly, in direct contrast to convergence factors, total time to solution actually increases as $\epsilon_R$ decreases, in particular
due to an increase in setup time. Note in Figure \ref{fig:soc_r_study_setup_versus_total_time} that the solve time (difference between
total and setup) decreases slowly with smaller $\epsilon_R$, but the setup time increase dramatically. Trends are similar for both the
uniform mesh and unstructured mesh. The increase in total time for pAIR as $\epsilon_{R}$ decreases is more drastic for distance-2
$\ell$AIR than distance-1 $\ell$AIR, but trends are similar in both cases.
Clearly, for the scaling tests presented Section \ref{sec:weak_scaling_parallel_in_space},
the largest value of $\epsilon_{R}$ that still provides for reasonable convergence should be used.

\begin{figure}[!htb]
	\centering
	\includegraphics[width=\textwidth]{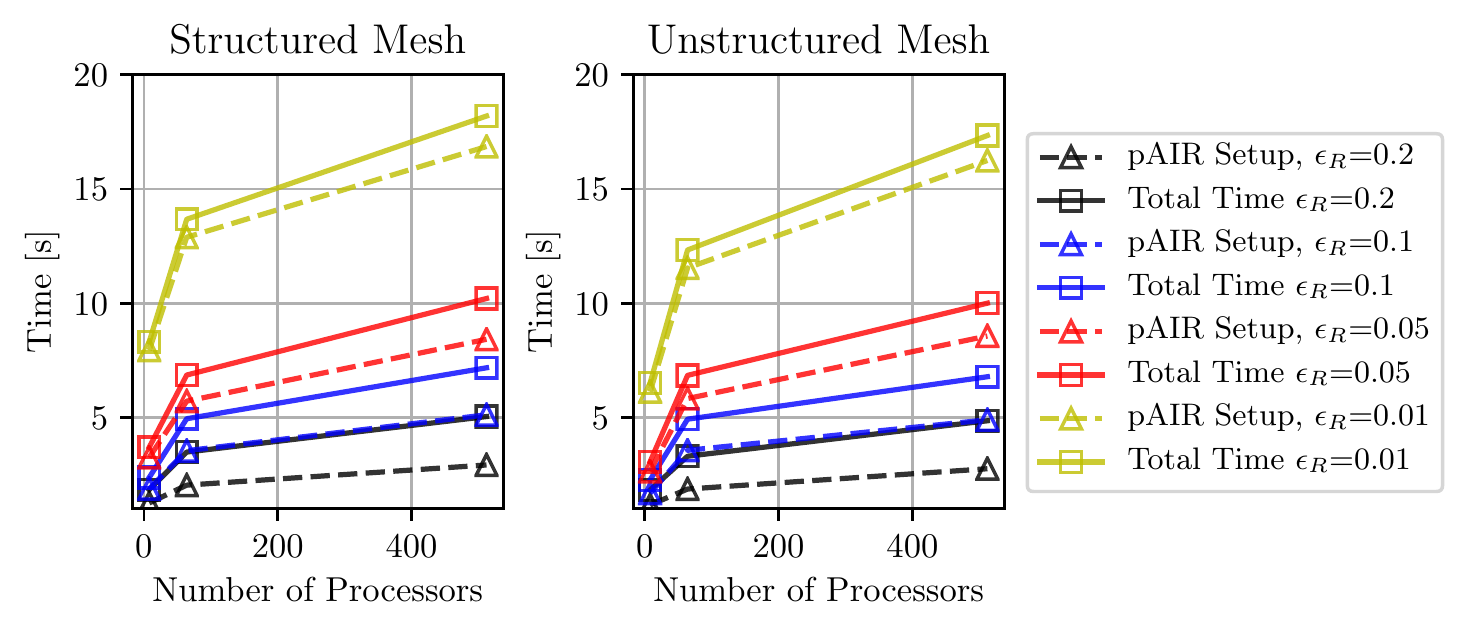}
	\caption{Comparison of total time to solution with setup time for different $\epsilon_R$ parameter values.}
	\label{fig:soc_r_study_setup_versus_total_time}
\end{figure}

\subsubsection{Stencil Growth and Filtering of $R$}

Ideally as problem size increases, the size (in memory) or ``complexity'' of the AMG hierarchy relative to initial problems size would
remain constant, and the average convergence factor would also remain constant. Classical AMG algorithms can exhibit these
qualities for some problems such as certain 2-dimensional discretizations of elliptic PDE problems. However, in 3-dimensions,
even these PDE problems may exhibit an increasing complexity as the mesh is refined. Complexity can be analyzed in several
different ways. One indicator of overall complexity is growth in average stencil size as the mesh is refined.
This is equivalent to examining the average number of nonzeros per row for each operator in the AMG hierarchy.
For all problems examined in this paper, stencil growth is apparent for a variety of parameter settings.

One strategy for limiting stencil growth is filtering small coefficients from operators. Filtering of matrices in the AIR hierarchy
is discussed in \cite{AIR2} where is it pointed out that while filtering of values from symmetric matrices is a sensitive procedure,
a system arising from a hyperbolic problem should be less sensitive to filtering. Here we extend that notion to introduce a
filtering on $R$. This is analogous to the pre- and post-filtering done for interpolation sparsity patterns in the root-node AMG
algorithm \cite{Manteuffel:2017}. Here, a SOC based on $\epsilon_R$ is used to determine an initial sparsity pattern, then $R$ is
constructed. Once $R$ is constructed, an analogous SOC as in Equation \eqref{eq:soc} is applied using a new tolerance, $\phi_R$,
and non-strong (i.e., weak) entries are eliminated from $R$. 

Figure~\ref{fig:filter_A_versus_R_for_distance_1} shows the significant reduction filtering in $R$ has on stencil growth for a
representative three-dimensional problem, and Figure~\ref{fig:3d_scaling_best_result} shows scaling results where small
coefficients have been filtered from $R$ after construction. The pAIR solve is still robust even with large filtering parameter
values, and the time to solution and memory consumption drop accordingly. 

\begin{figure}[!htb]
	\centering
	\includegraphics[width=\textwidth]{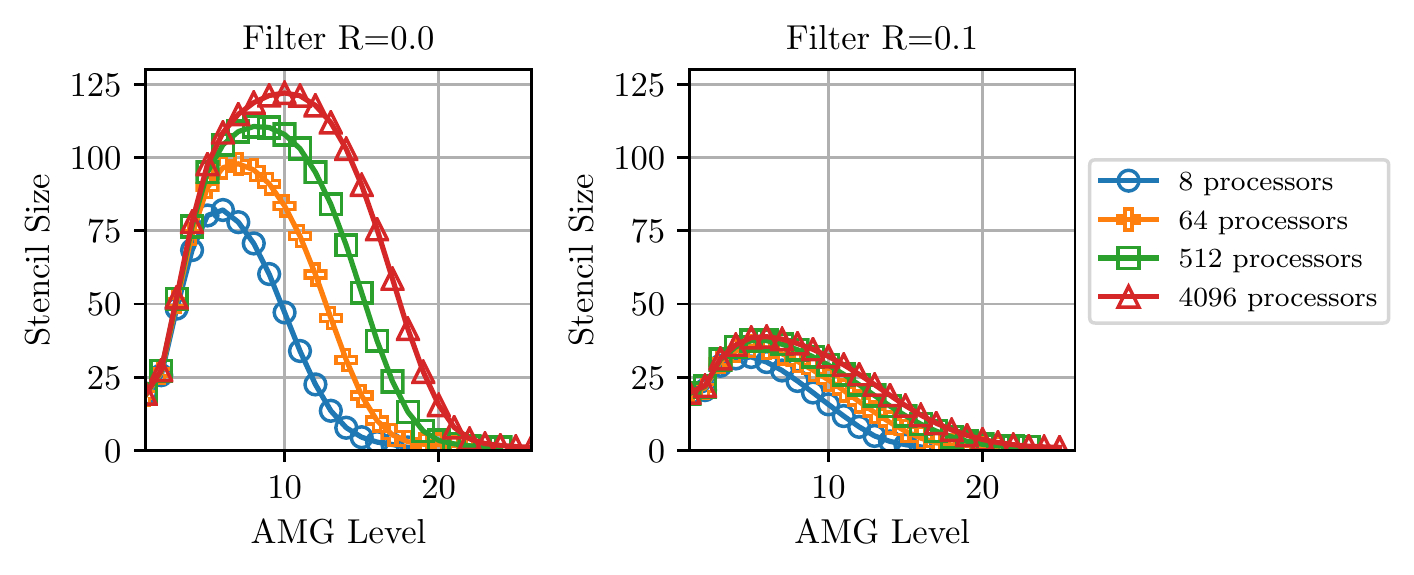}
	\caption{Comparison of aggressive filtering from $R$ to no filtering shown with distance-1 $\ell$AIR for a representative three-dimensional problem}
	\label{fig:filter_A_versus_R_for_distance_1}
\end{figure}

\begin{figure}[!htb]
	\centering
	\includegraphics[scale=0.9]{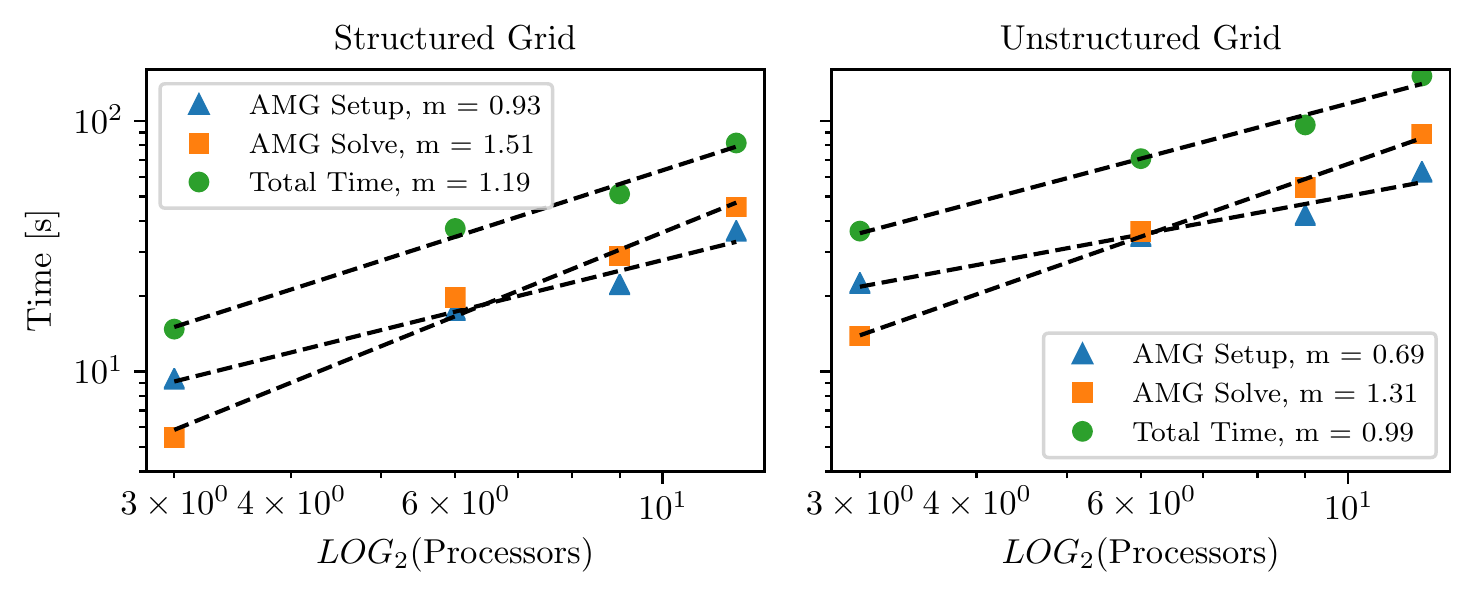}
	\caption{Scaling results for both a uniform grid and unstructured mesh for distance-1 nAIR with aggressive operator filtering using a filter parameter of $\phi_R = 0.1$. Time reported is that to solve all 32 directions in the angular quadrature set.}   
	\label{fig:3d_scaling_best_result}
\end{figure}

\section{Parallel in Angle and Minimizing Communication}
\label{sec:multiple_meshes}

In highly parallel environments, whether the transport sweeps are solved using pAIR or a traditional parallel
sweep, communication is often the dominant cost, particularly when considering strong scaling. This will be especially 
true on emerging-type architectures where nodes have GPUs or manycore processors that are particularly fast
on threaded, shared-memory-type floating-point operations. To reduce communication cost, we propose
including parallelism in the angular and energy domains, in addition to the spatial domain, by storing multiple 
copies of the spatial discretization. Multiple angles/energies can then be solved for simultaneously, and the
overhead communication cost to communicate angular and energy solutions is small compared to the communication
cost of spatial solves. Note that analogous ideas have been considered in the context of traditional transport
sweeps as well, originally in \cite{Dorr:1993hx,Dorr:1996ec}.

This section develops a model for the number of parallel communication stages of pAIR in terms of number of
processors, depending on how many copies of the spatial problem are used. In particular, it is shown that the
minimum number of stages of communication is achieved by using the
maximum number of spatial meshes that can fit in memory. Although this does not account for computation
time, when using pAIR it is a reasonable proxy for computation time as well. When using an on-processor 
solve for relaxation with pAIR, the larger the portion of spatial domain that is stored on processor, the faster
convergence will be, with the obvious limit of a forward (exact) solve in one iteration when the spatial domain is stored on
one processor. The performance model is presented in Section \ref{sec:multiple_mesh_theory}, followed by
numerical results in Section \ref{sec:multiple_mesh_results}. For ease of readability, proofs for theoretical results
in Section \ref{sec:multiple_mesh_theory} are left to the appendix. 

\subsection{A Performance Model}
\label{sec:multiple_mesh_theory}

This section develops a simple performance model for stages of communication in a pAIR transport sweep
as a function of the number of copies of the spatial discretization. The basic result says that the minimum
number of stages of communication is obtained with the maximum number of spatial meshes that can fit
in memory.

Let $N$ denote the total spatial degrees-of-freedom (DOFs), $M$ the total angular DOFs, $E$ the total energy
DOFs, $P$ the number of processors, and $C_P$ the memory capacity of each processor. We are solving
for the angular/energy flux, $\psi_{\ell,\nu}(\textbf{x})$, where $\ell$ denotes the angle, $\nu$ the energy, and $\textbf{x}$
the spatial point. Thus, $\psi_{\ell,\nu}(\textbf{x})$ has dimension $NME$. We must also account for the memory
to store pAIR hierarchies, where one hierarchy is typically equivalent to storing 2--4 copies of the matrix for a
single angle. Because memory is a leading constraint in transport simulations, here we only consider the
case of storing one hierarchy at a time and, thus, rebuilding a hierarchy for every solve. Following from this
discussion, assume that $P$ processors have sufficient storage for this problem.

\begin{assumption}\label{ass2}
	Assume 
	\begin{equation} \label{eq:Ass2}
	PC_P \geq  NC_T(ME + C_{A}),
	\end{equation}
	where $C_T$ is a factor that describes the memory required to store each spatial DOF, plus any auxiliary
	vectors, $C_{A}$ a factor that describes the memory to store an AMG hierarchy for a single angle, relative
	to the storage for all spatial DOFs (should be $\approx$ 2--4).
\end{assumption}

Then, divide the machine into $K$ equal partitions, each of which includes a full spatial discretization,
and distribute the $ME$ angle and energy DOFs across the $K$ partitions. That is,
each spatial partition addresses $M_1 = M/K_A$ angles and $E_1 = E/K_E$ energy groups,
where $K = K_A K_E$. This implies $M_1 E_1 = ME/K$.  We further
assume that there are not more partitions than angle/energy DOFs. 

\begin{assumption}\label{ass:ME}
	Assume that each angle/energy DOF is associated with only one partition, that is, $K \leq ME.$
\end{assumption}
In the remainder of this section, we assume that all $M_1 E_1$ angle/energy DOFs associated with
a spatial node are stored on the same processor. This simplifies the accumulation of scalar flux
and the scattering source. 

The size of $K$ is further restricted by the requirement that the capacity of $P_x$ processors will
accommodate all $M_1E_1$ local DOFs and an AMG hierarchy for a single angle/energy. This
leads to the following constraint:
\begin{constraint}\label{const2}
	\begin{align}
	P_x C_P & \geq N C_T(M_1E_1 + C_A) \hspace{5ex}\Longleftrightarrow\hspace{5ex}
	K \leq \frac{PC_P- NC_TME}{C_aNC_T}.
	\end{align}
\end{constraint}
Note that Assumption \ref{ass2} implies Constraint \ref{const2} with $K = 1$. We can now state the
primary results.

\begin{theorem}\label{th:optK}
	Assume Constraint \ref{const2} and Assumption \ref{ass:ME} hold and that pAIR scales like $\mathcal{O}(\log P)$.
	Then the optimal $K$ with respect to minimizing the number of communication stages is given by
	\begin{align*}
	K_{opt} = \min\left\{ P, ME, \frac{PC_P- NC_TME}{C_aNC_T}\right\}.
	\end{align*}
\end{theorem}
\begin{conjecture}\label{conj:optK}
	Assume Constraint \ref{const2} and Assumption \ref{ass:ME} hold and that pAIR scales like $\mathcal{O}(\log P)^\alpha$,
	for $\alpha\in(1,2]$. Then the optimal $K$ with respect to minimizing the number of communication stages is given by
	\begin{align*}
	K_{opt} = \min\left\{ P, ME, \frac{PC_P- NC_TME}{C_aNC_T}\right\}.
	\end{align*}
\end{conjecture}

Conceptually, these results say that minimizing the number of stages of communication is obtained by using the maximum
number of copies of the spatial mesh that fit in memory. Of course this doesn't account for more complex aspects of
performance, including computation time, changes in memory access time and message size as the local spatial problem
size changes, etc. However, it does suggest that using multiple spatial meshes should improve performance, a result
that is consistent with numerical tests in the following subsection. 

Note, Conjecture \ref{conj:optK} is not proven (or disproven) to hold for all situations, but a partial proof in the Appendix
shows that for moderate to large $P$, it is very likely to hold, and can be verified for specific values of
the above parameters.

\subsection{Parallel in Angle Results}
\label{sec:multiple_mesh_results}

\begin{figure}[!htb]
	\centering
	\includegraphics[scale=0.35]{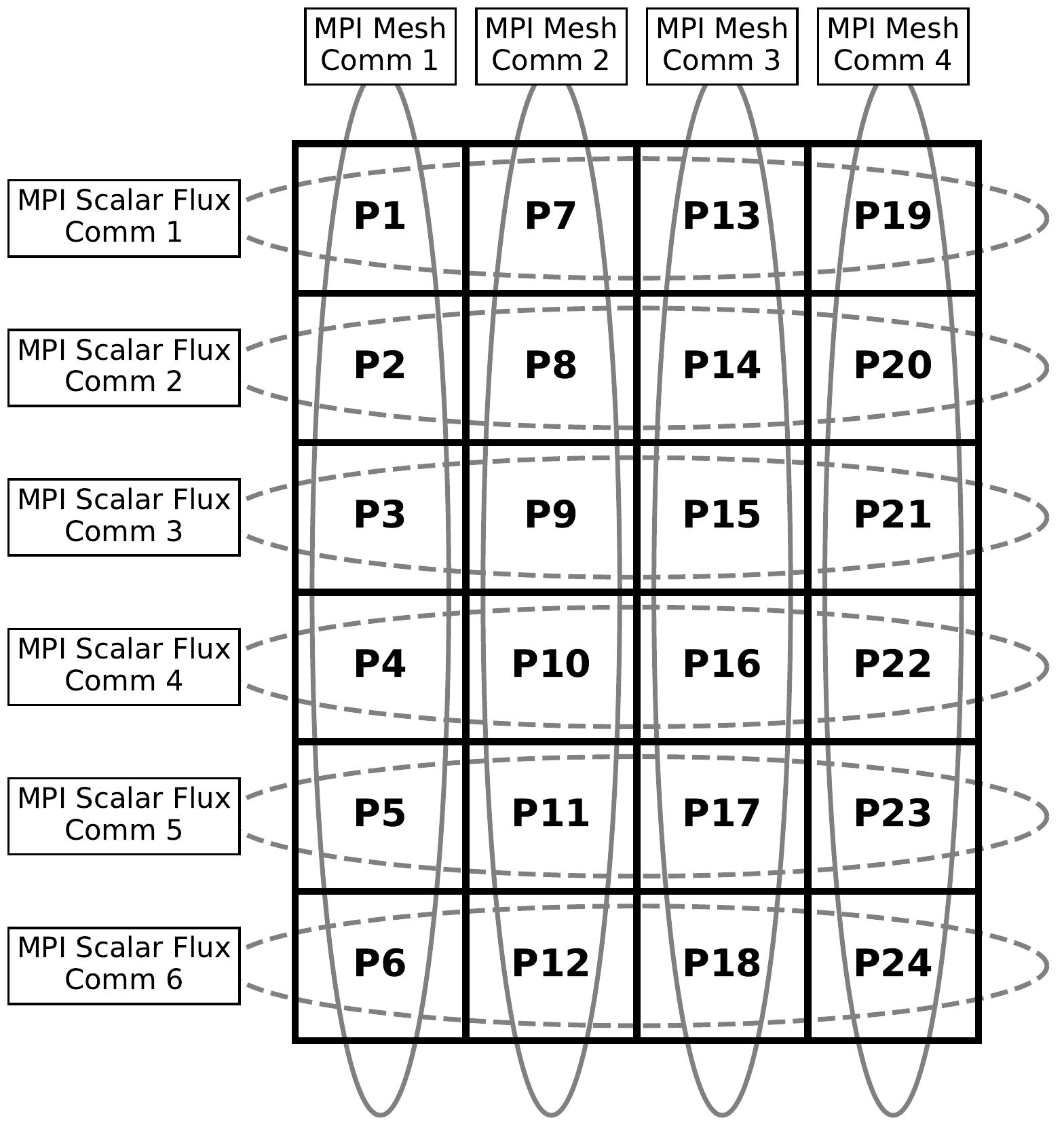}
	\caption{Schematic showing the grouping of 24 processors into a grid of $6\times 4$ MPI communicators. Columns each contain
	an identical copy of a distributed spatial mesh to solve for angular flux vectors over a subset of directions in the angular discretization.
	Rows each contain a full angular discretization over a fixed subdomain of the spatial mesh.}
	\label{fig:multi_mesh_schematic_1}
\end{figure}

Figure \ref{fig:multi_mesh_schematic_1} provides an example of the parallel in angle implementation
on 24 processors, where it is assumed there is sufficient memory to store the spatial mesh and associated variables on
six processors. Six MPI groups and associated intra-communicators are thus created and single direction solves are
performed on the processors in these groups, each group accounting for the angular flux directions associated with 1/4 of the $M$
directions in the S$_N$ discretization. To form the scalar flux as a weighted sum over angular flux, these intra-communicators,
called MPI Mesh communicators in Figure \ref{fig:multi_mesh_schematic_1}, must then perform an Allreduce communication,
where the contribution from each Mesh communicator to the scalar flux is summed and this sum is made available to all processors.
To perform this Allreduce, a second MPI grouping is created, denoted as the Scalar Flux communicators in Figure~\ref{fig:multi_mesh_schematic_1}.
Note the spatial mesh must be distributed in exactly the same way based on MPI rank in each MPI Mesh communicator. This is easily
accomplished with the software libraries used here.

The weak scaling tests using a three-dimensional structured mesh described in Section~\ref{sec:weak_scaling_parallel_in_space} are now repeated for the parallel in
space and angle implementation. Three different parallel in angle simulations are performed using two, four, and eight
replicated meshes respectively, with an on-processor solve for relaxation. 
Results from four specific configurations are presented in this section, two using $\ell$AIR and two using nAIR, and
with two values of filtering for $R$. 
These specific cases are presented because there is significant difference between the D-2 $\ell$AIR algorithm and the degree-1
nAIR algorithm.
Additionally, with heavy operator filtering, stencil growth is minimized along with the overall communication costs.
Despite this, significant speed up is seen at 4096 processors with more mesh
replication, particularly in setup time, with reductions in total time per sweep (including pAIR setup and solve for all angles) 
reduced by 30-50\% compared with no parallelism in angle. The pAIR setup times are shown in Figure
\ref{fig:scaling_multi_mesh_setup_timing}.

\begin{figure}[!htb]
	\centering
	\includegraphics[width=\textwidth]{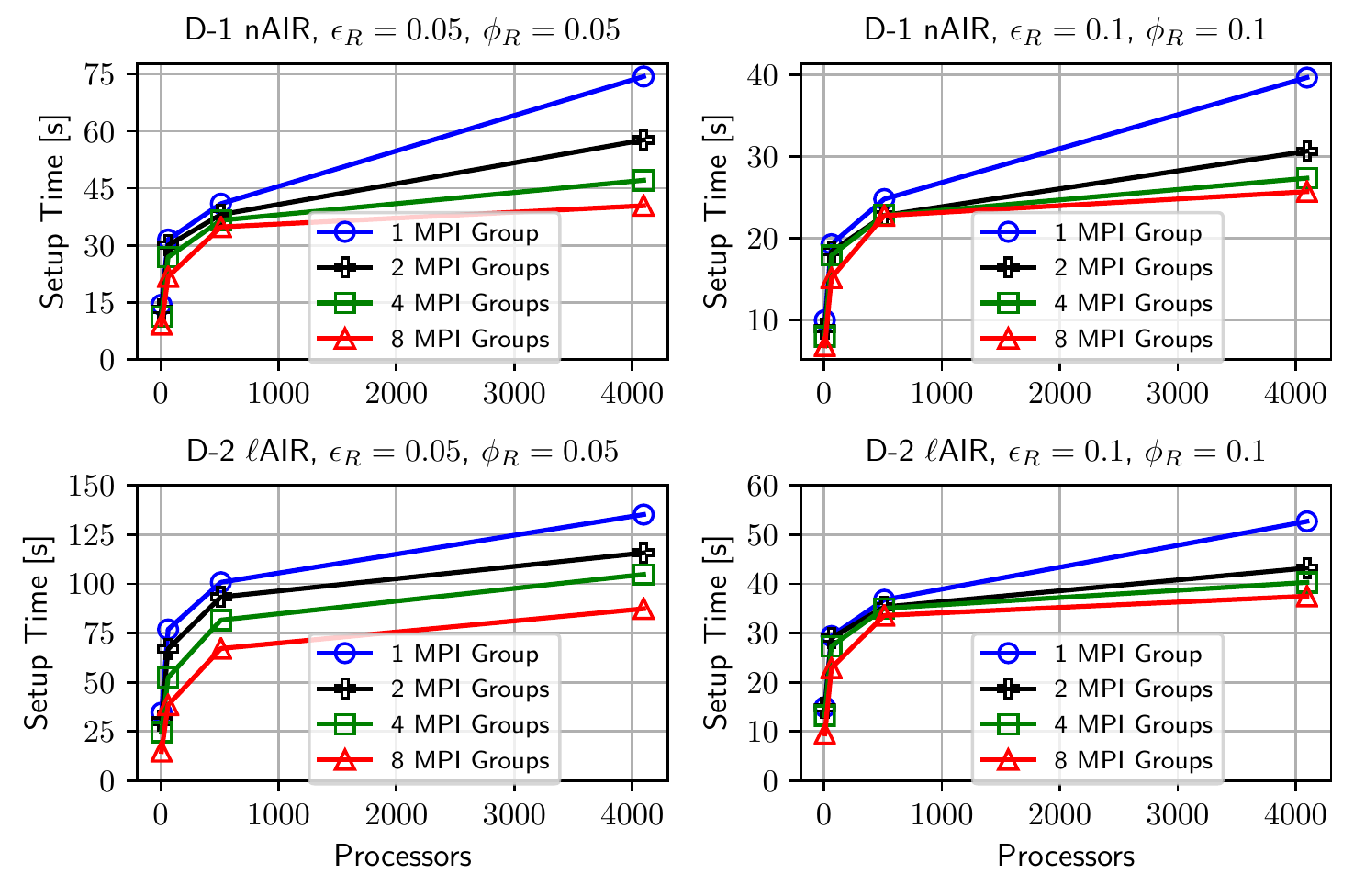}
	\caption{AMG setup time and solve time shown for weak scaling with different numbers of replicated meshes. Results are shown for both $\ell$AIR and nAIR. $\phi_R$ is the level of filtering on $R$}
	\label{fig:scaling_multi_mesh_setup_timing}
\end{figure}

The solve time and average convergence factor, for D-1 nAIR are shown in Figure
\ref{fig:scaling_multi_mesh_avg_convergence_factors_solve_time}. Results are only shown for nAIR because
the solve time and CF of D1 nAIR and D-2 $\ell$AIR are almost identical (in fact, the preconditioners are almost
identical when applied to linear transport on a structured grid \cite{AIR2}). Here, the linear systems are solved to a
relative tolerance of 1e-6 instead of to machine precision as in Section~\ref{sec:weak_scaling_parallel_in_space}.
Notice that the average convergence factors decrease significantly with multiple meshes, due to the on-processor
relaxation solving a larger part of the domain directly, but the solve times are more or less constant. In this case,
the reduced communication and faster convergence is offset by the additional computation required when a larger
portion of the spatial problem is stored on processor.

\begin{figure}[!htb]
	\centering
	\includegraphics[width=\textwidth]{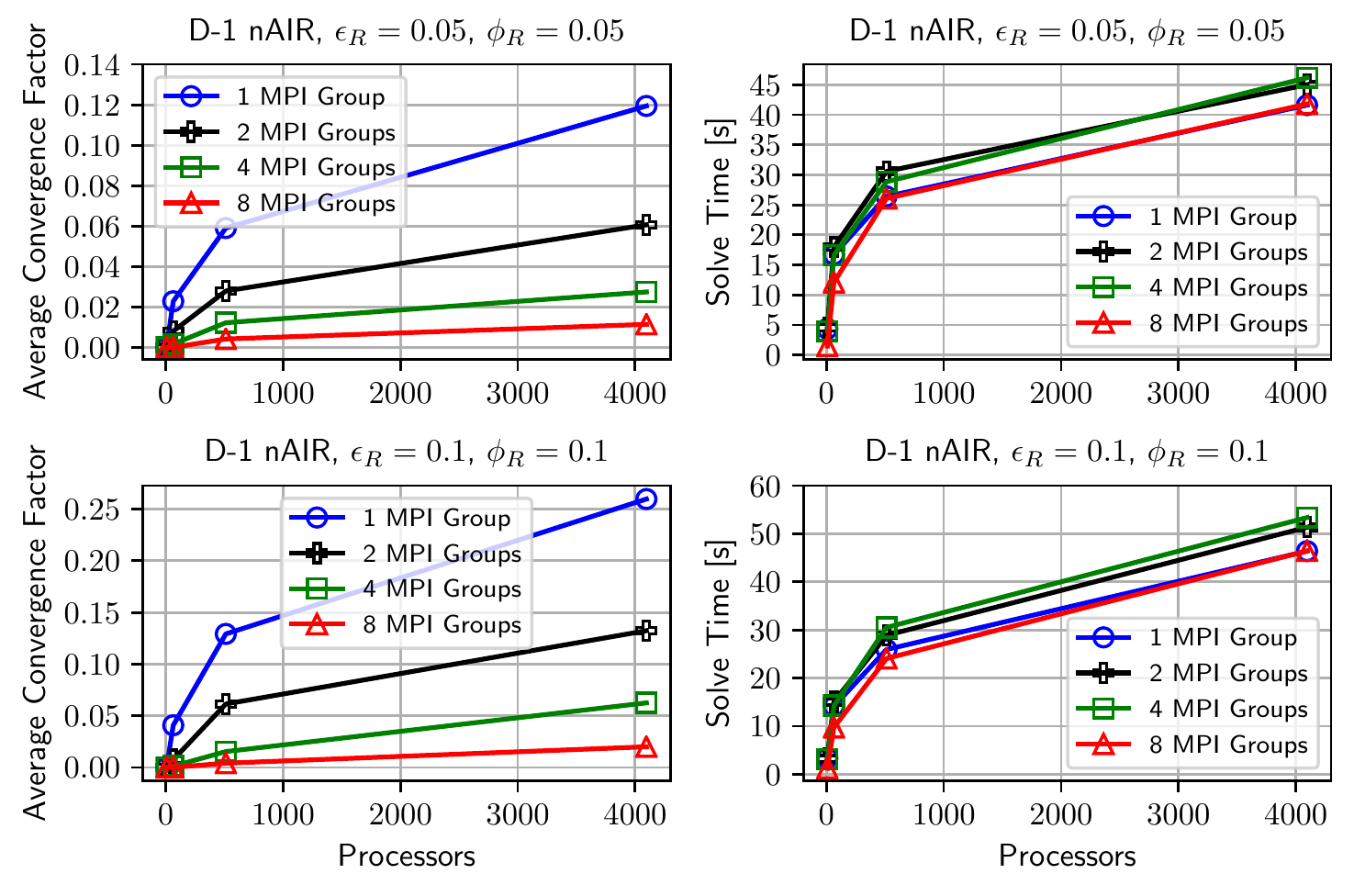}
	\caption{Average convergence factors shown for weak scaling with different numbers of replicated meshes. Results are shown for both $\ell$AIR and nAIR}
	\label{fig:scaling_multi_mesh_avg_convergence_factors_solve_time}
\end{figure}

Although these results only consider the setup and solve time for a single sweep, they extend to the general time to
solution of source iteration as well.
The only overhead of parallel in angle vs. only spatial parallelism is additional communication to compute the scalar
flux. However, communicating and computing the scalar flux is very cheap compared to angular flux
calculations with pAIR. For example, with 8 MPI groups the computation and communication to form the scalar flux
takes less than 1/10th of a second on 4096 processors, an insignificant amount of time compared to the reduction
in time obtained through mesh replication. Thus, if multiple meshes reduce the total time for a single iteration by 30\%,
then total time of $n$ source iterations will also decrease by $\approx 30\%$.

\section{Representative Parallel Sweep Results}
\label{sec:comp_with_sweeps}
As discussed in Section~\ref{sec:sn_by_parallel_sweeps}, provably optimal methods have been developed for solving the $S_N$ equations discretized on a structured mesh using parallel sweeps.
These methods have been implemented in the deterministic transport program PDT \cite{adams1,Tanase:2011:SPC:2038037.1941586}.
Figure~\ref{fig:pdt_scaling} shows the time required by PDT to complete a full sweep of all 32 directions used in the three-dimensional problem with a uniform mesh described in Section~\ref{sec:weak_scaling_parallel_in_space}. 
Results are presented for a hybrid KBA partitioning as well as a volumetric partitioning.
The volumetric partitioning has better parallel scalability than the hybrid KBA, but there is a larger constant associated with the scaling law \cite{Pautz_sweep_unstructured_mesh_decomp} and so hybrid KBA is less costly for the problem investigated here. 

When presenting scaling results, it is common and sensible to discuss relative times; however, actual solution times are presented in this section.
The purpose of these results is not to investigate the scalability of parallel sweeps, as has been done in other works including \cite{Validation_of_parallel_sweeps}, but to provide a representative result for the cost of parallel sweeps for comparison.
The times presented depend on many factors besides the parallel algorithm including details of code implementation and computer characteristics.

\begin{figure}[!htb]
	\centering
	\includegraphics{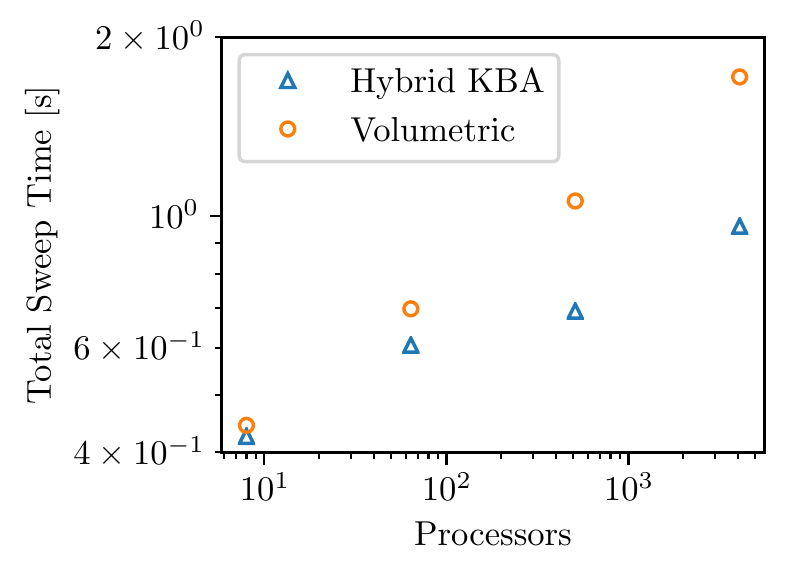}
	\caption{Time taken for a single parallel sweep of 32 directions with PDT. Results are shown for both a hybrid KBA and volumetric cell partitioning}   
	\label{fig:pdt_scaling}
\end{figure}

As discussed in Sections~\ref{sec:sn_by_parallel_sweeps} and \ref{sec:sweeping_with_pAIR} parallel sweeps have poorer
theoretical scaling compared to AMG, which has a logarithmic dependence on $P$ rather than polynomial; however, these scaling laws say
nothing about the constants associated with the growth in solution time.
The constant associated with the AMG weak scaling law is generally much larger than that associated with parallel sweeps for a structured mesh.
For the three-dimensional cube investigated here with a structured mesh, fitting a constant $CP^{1/3}$ to the data in Figure \ref{fig:pdt_scaling}
and extrapolating to large $P$ predicts a crossover point where pAIR would out perform parallel sweeps on the order of 100 million processors.

Despite the large crossover point, it is worth pointing out that PDT is highly optimized and designed for uniform structured grids.
These results will not be generally applicable to unstructured meshes. It is challenging to make a general statement regarding the efficiency
of parallel sweeps on unstructured meshes. For example, a recent work has shown good scaling for parallel sweeps out to millions of
MPI processes when there is some coarse regularity to an unstructured mesh such that it can be decomposed into well balanced
regular partitions \cite{2019arXiv190602950A}. However, to our knowledge, no comparably scalable method has been demonstrated
for a general unstructured mesh, even in the case where the mesh has no cycles. Furthermore, cycles in a mesh (in which case
the resulting matrix is not triangular) require determining some ordering in which to sweep. In \cite{haut2019efficient}, a cycle-breaking
strategy is developed for highly curved meshes. That strategy provides an important framework in which to sweep on very unstructured
or curvilinear meshes, but source iteration was shown to converge in up to $3\times$ less iterations when inverting the transport 
equations exactly with pAIR compared with sweeping and cycle breaking. 

Conversely, pAIR is robust on arbitrary unstructured and curvilinear meshes, but optimization of pAIR is an ongoing process.
Preliminary data suggests setup time of pAIR can be reduced significantly, and
recent work on parallel sparse matrix operations \cite{bienz2019reducing,bienz2019node} shows that improved communication
algorithms can speed up the setup and solve phase by several times. Section \ref{sec:accuracy} indicates that pAIR solves do
not need to be performed to a high accuracy. If we consider time-dependent transport (easier to solve than steady state),
only a few pAIR iterations should be necessary, which will further reduce pAIR solve time by several times. Combining 
superior performance on unstructured or curvilinear meshes, a reduction in iterations for the time-dependent setting, and
the potential reduction in total pAIR time to solution by several times, a crossover point may be possible in the millions
of processors for some problems.

\section{Conclusion}
\label{sec:conclusion}

In this paper, pAIR is shown to be effective and scalable for solving the source iteration equations of the S$_N$ approximation to
the transport equation. Section~\ref{sec:perf:mesh} demonstrates pAIR is capable of excellent convergence factors for a variety of
meshes and material configurations, and scaling tests in Section~\ref{sec:weak_scaling_parallel_in_space} show pAIR as currently
implemented is capable of nearly ideal multigrid scaling.

When memory allows for some number of replicated meshes, Section \ref{sec:multiple_meshes} shows that a parallel in angle
with mesh replication implementation can significantly increase performance. In a multiphysics code, using this multiple mesh
scheme globally may not be possible. However in this situation, if there is excess memory and a relatively fast mesh partitioning
algorithm is available, a separate replicated mesh partitioning might be constructed and used for the transport solve. 

Because pAIR is an algebraic solver implemented in a popular linear algebra library, it can be easily interfaced as a black-box
solver to other programs. As demonstrated here, a program capable of solving transport was created with a popular FEM library.
In particular, no changes or special treatment were required when using the p4est domain decomposition library to distribute and
refine the mesh, and AIR has proven effective on arbitrary unstructured and curvilinear meshes. This is not generally possible
to do with parallel sweeps, which even on structured grids, require specialized communication scheduling logic that is
integral with domain partitioning. Although the superior asymptotic scaling of pAIR over traditional sweeps appears insufficient 
for pAIR to overtake sweeps in performance, at least on meshes with moderate structure, the generality and non-intrusiveness
of pAIR as a solver offers significant other advantages. 

\section*{Appendix}
\label{appendix}

The MMS source, $q_{mms}$, for error tests in Section \ref{sec:accuracy} is shown below,
where $\mu$, $\eta$, and $\zeta$ are the direction cosines and L is the length the cube side.

\begin{align*}
q_{mms}&=\langle \mu, \eta, \zeta \rangle \cdot \langle \frac{\partial \left( \Phi \Psi \right)}{\partial x}, \frac{\partial \left( \Phi \Psi \right)}{\partial y}, \frac{\partial \left( \Phi \Psi \right)}{\partial z} \rangle - \frac{ \sigma_s}{4\pi} \Phi +\sigma_t \Phi \Psi \\
\Phi&=\sin\left( \frac{x\pi}{L} \right) \sin\left( \frac{y\pi}{L}  \right) \sin\left( \frac{z\pi}{L}  \right) \\
\Psi&=\left(1+2\mu \left( x-\frac{L}{2} \right)  \right)  \left(1+2\eta \left( y-\frac{L}{2} \right)  \right) \left(1+2\zeta \left( z-\frac{L}{2} \right)  \right)
\end{align*}

\begin{proof}[Proof of Theorem \ref{th:optK}]
The proof consists of deriving a function $S_{AMG}(K)$ that gives the total number of communication
stages as a function of $K$, and proceeding to show that $\tfrac{\partial S_{AMG}}{\partial K} < 0$ for
$K \leq \min\{P, ME\}$. Thus, increasing $K$ to the minimum of these values is guaranteed to reduce
$S_{AMG}$. By Assumption \ref{ass:ME} and the fact that we cannot have more partitions than
processors $P$, we cannot pick $K$ larger than the minimum of these values.
If $\tfrac{PC_P- NC_TME}{C_aNC_T} < \min\{ME,P\}$, then Constraint \ref{const2} forces us to pick
$K = \tfrac{PC_P- NC_TME}{C_aNC_T}$. 

Consider a model of total communication stages when applying pAIR as a solver for the transport
sweeps. AMG has two phases, the setup and solve phase, with a total number of communication
stages given by
\begin{align}\label{eq:amg_stage}
\xi\log_2\left(\frac{P}{K}\right)^{\alpha},
\end{align}
for $\alpha \in[1,2]$. Optimal performance is given by $\alpha = 1$ (the assumption of
the theorem, but we stay general for purposes of the conjecture). 
Here, $\xi$ is a constant that reflects the cost of the setup and solve of AMG compared
with a sweep, and typically $\alpha\in[1,1.5]$. Each partition must solve $ME/K$ angles,
so \eqref{eq:amg_stage} is multiplied by $ME/K$. 

In addition to sweeps, the scalar flux must be accumulated for each energy, and the scattering
source must be computed over all energy groups. These steps are accomplished by computing
the scalar flux available on each processor and summing across partitions so that all partitions
have the scalar flux for the energies on that partition.  Let $\psi_{\ell,\nu}(\textbf{x})$ represent the
angular flux in direction $\ell$ with energy $\nu$. Let $\phi_{\nu}(\textbf{x})$ be the scalar flux with
energy $\nu$.  Let $w_\ell$ be the quadrature weight and let $\sigma_{\nu,\nu'}$ be the scattering
cross-section from energy group $\nu'$ to energy group $\nu$. The accumulation of the scalar
flux has the form
\begin{equation*}
\phi_\nu ({\textbf{x}}) = \sum_1^{K_A} \left(\sum_{\ell= 1}^{M_1} w_\ell \psi_{\ell,\nu}(\textbf{x}) \right),
\end{equation*}
where the second sum represents summation on each processor and the first sum is across the $K$ 
partitions. Likewise, accumulation of the scattering term has the form
\begin{equation*}
Q_{\nu} ({\textbf{x}})= \sum_1^{K_E} \left( \sum_{\nu'=1}^{E_1} \sigma_{(\nu,\nu')} \phi_{\nu'}({\textbf{x}})\right).
\end{equation*}

This requires $S_A = \log_2(K_A)$ stages to compute all the scalar fluxes for each energy and 
$S_E = \log_2(K_E)$ stages to compute the scattering kernel.
Combining with the communication for pAIR leads to a model of the form
\begin{align*}
S_{AMG} & =  \xi\log_2\left(\frac{P}{K}\right)\frac{ME}{K} + \zeta\log_2(K_E) + \zeta\log_2(K_A) \\
& = \xi\log_2\left(\frac{P}{K}\right)\frac{ME}{K} + \zeta\log_2(K).
\end{align*}
Here, a constant $\zeta$ is introduced to acknowledge that the summation stages require less
communication than the pAIR stages. Thus, we expect that $\zeta \ll \xi$.\footnote{This could
also be normalized, but we prefer to explicitly include some leading constant for each.}

Then,
\begin{align}
\frac{\partial S_{AMG}}{\partial K} & = \frac{\zeta}{\ln(2)K} - \frac{\xi ME }{\ln(2)K^2}
	- \frac{\xi ME \log_2\left(\frac{P}{K}\right)}{K^2} \nonumber\\
& = \frac{1}{K^2} \left[ \frac{K\zeta}{\ln(2)} - \xi ME \left(\frac{1}{\ln(2)} +
	\log_2\left(\frac{P}{K}\right)\right)\right].\label{eq:dsamg}
\end{align}
Note that multiplying $\tfrac{\partial S_{AMG}}{\partial K}$ by $K^2$ does not change the
sign of the slope or the zeros, so consider the interior term
\begin{align}
\mathcal{G}_0(K) & := \frac{\zeta}{\ln(2)} K - \frac{\xi ME}{\ln(2)}\left(1+\ln\left(\frac{P}{K}\right)\right), \label{eq:dG0dk} \\
\frac{\partial \mathcal{G}_0}{\partial K} & = \frac{\zeta}{\ln(2)} + \frac{\xi ME}{\ln(2) K}.\nonumber
\end{align}
Note that $\mathcal{G}_0(1) < 0$ and $\frac{\partial \mathcal{G}_0}{\partial K} > 0$ for $K > 0$. Thus there exists
exactly one zero of $\frac{\partial S_{AMG}}{\partial K}$ for $K > 1$. A closed form in terms of special functions
is given in the following lemma. 

\begin{lemma}\label{lem:W}
For $\alpha = 1$, the single zero of $\frac{\partial S_{AMG}}{\partial K}$ \eqref{eq:dsamg} is given by
\begin{align*}
K_0 & = \frac{\xi ME}{\zeta} W\left(\frac{e\zeta P}{\xi ME}\right),
\end{align*}
where $W(\cdot)$ denotes the Lambert $W$-function.
\end{lemma}
\begin{proof}
The proof proceeds by construction. Recall the identity of the Lambert $W$-function that $\ln(W(x)) = \ln(x) - W(x)$
\cite{corless1996lambertw}. Let $C = \tfrac{\xi ME}{\zeta}$ and observe
\begin{align*}
C + C\ln\left(\tfrac{P}{K_0}\right)  & = C+  C\ln(P) - C\ln(K_0) \\
& = C + C\ln(P) - C\ln\left(CW\left(\tfrac{eP}{C}\right)\right) \\
& = C + C\ln(P) - C\ln(C) - C\ln\left(\tfrac{eP}{C}\right) + CW\left(\tfrac{eP}{C}\right) \\
& = CW\left(\tfrac{eP}{C}\right) \\
& = K_0.
\end{align*}
Thus $K_0$ satisfies the relation 
\begin{align*}
K_0 & = \frac{\xi ME}{\zeta}\left(1 + \ln\left(\frac{P}{K}\right)\right),
\end{align*}
Appealing to \eqref{eq:dG0dk} completes the proof.
\end{proof}

We now have everything we need to prove the final result.
Suppose $P < ME$ and let $K = P$. Then from \eqref{eq:dG0dk},
\begin{align*}
G_0(P) = \frac{\zeta P}{\ln(2)} - \frac{\xi ME}{\ln(2)} < 0.
\end{align*}
Because $G_0(1) < 0$ and $\tfrac{\partial \mathcal{G}_0}{\partial K} > 0$, this implies that
$\tfrac{\partial S_{AMG}}{\partial K} < 0$ for $K \in[1,\min\{P,ME\}] = [1,P]$.

Now suppose $P \geq ME$. Appealing to Lemma \ref{lem:W}, it is sufficient to show that the root $K_0$ of
$\tfrac{\partial S_{AMG}}{\partial K}$ satisfies $K_0 > \min\{P, ME\} = ME$, which implies
$\tfrac{\partial S_{AMG}}{\partial K} < 0$ for $K \in[1,\min\{P,ME\}] = [1,ME]$. Recall that for real, positive
values, $W(x)$ is monotonically increasing \cite{corless1996lambertw}, as well as the definition of the
Lambert $W$-function given by $x = f^{-1}(xe^x):=W(x)$. Then observe,
\begin{align*}
P\geq ME \hspace{3ex}\implies\hspace{10.5ex} \frac{P}{ME} & \geq e^{\tfrac{\zeta}{\xi} - 1},  \\
\Longleftrightarrow\hspace{10ex} \frac{e\zeta P}{\xi ME} & \geq \frac{\zeta}{\xi} e^{\tfrac{\zeta}{\xi}} , \\
\implies\hspace{3ex} W\left(\frac{e\zeta P}{\xi ME} \right) & \geq W\left(\frac{\zeta}{\xi} e^{\tfrac{\zeta}{\xi}}\right),  \\
\Longleftrightarrow\hspace{3.5ex} W\left(\frac{e\zeta P}{\xi ME} \right) & \geq \frac{\zeta}{\xi},  \\
\implies \hspace{4.25ex} K_0 \geq ME.
\end{align*}

\end{proof}

\begin{proof}[Partial proof of Conjecture \ref{conj:optK}]
For $\alpha\in(1,2]$, $S_{AMG}$ take the more general form
\begin{align*}
S_{AMG} & = \xi\log_2\left(\frac{P}{K}\right)^{\alpha}\frac{ME}{K} + \zeta\log_2(K), \\
\frac{\partial S_{AMG}}{\partial K} & = \frac{1}{K^2} \left[ \frac{K\zeta}{\ln(2)} -
	\xi ME \left(\frac{\alpha\log_2\left(\frac{P}{K}\right)^{\alpha-1}}{\ln(2)} +
	\log_2\left(\frac{P}{K}\right)^\alpha\right)\right].\label{eq:dsamg}
\end{align*}
As before, we multiply $\tfrac{\partial S_{AMG}}{\partial K}$ by $K^2$ and define the functional
\begin{align*}
\mathcal{G}(K) & = \frac{\zeta}{\ln(2)} K - \xi ME \left(\frac{\alpha\log_2\left(\frac{P}{K}\right)^{\alpha-1}}{\ln(2)} +
	\log_2\left(\frac{P}{K}\right)^\alpha\right),
\end{align*}
where 
\begin{align*}
\frac{\partial \mathcal{G}}{\partial K} & = \frac{\zeta}{\ln(2)} + \frac{\xi ME \alpha}{\ln(2)K} \left(\frac{(\alpha-1)\log_2\left(\frac{P}{K}\right)^{\alpha-2}}{\ln(2)} +
	\log_2\left(\frac{P}{K}\right)^{\alpha-1}\right).
\end{align*}
We assume that $K \in[1,P]$, in which case $P/K \geq 1$ and $\log_2\left(\tfrac{P}{K}\right) \geq 0$.
Because the remaining constants, $\xi, \zeta, M$, etc., are positive, $\tfrac{\partial \mathcal{G}}{\partial K} > 0$ for $K \in[1,P]$.
Suppose $P \geq 2$. Then it is clear for $K = 1$, $\tfrac{\partial S_{AMG}}{\partial K} < 0$. Conversely, if we let $K = P$,
the log-terms vanish and $\tfrac{\partial S_{AMG}}{\partial K} > 0$. Because $\tfrac{\partial \mathcal{G}}{\partial K} > 0$ for
$K \in[1,P]$, there exists exactly one zero of $\tfrac{\partial S_{AMG}}{\partial K}$ over the interval $K\in(1,P)$.

Determining the root of $\tfrac{\partial S_{AMG}}{\partial K}$ without explicit constants is not trivial.
However, we can note that 
\begin{align*}
\frac{\partial S_{AMG}}{\partial K}[P] = \frac{\zeta}{\ln(2)P} < \frac{1}{P}.
\end{align*}
Moreover, note that
\begin{align*}
\frac{\partial^2S_{AMG}}{\partial K^2} & = \frac{\mathcal{G}'(K)}{K^2} - \frac{2\mathcal{G}(K)}{K^3}.
\end{align*}
We will now show that if $\alpha\in(1,2)$, $\tfrac{\partial^2S_{AMG}}{\partial K^2} > 0$ over the range $(1,P)$, implying
that $S_{AMG}$ is concave up. The purpose of this is that if $\frac{\partial S_{AMG}}{\partial K}[P] < \frac{1}{P}$,
and the slope of $S_{AMG}$ is strictly increasing, then for moderate to large $P$,
\begin{enumerate}\itemsep0em 
\item the root $K_0$ of $\tfrac{\partial S_{AMG}}{\partial K}$ is likely quite close to $K = P$, and
\item choosing $K = \min\{P, ME\}$ results in at most a marginal increase in $S_{AMG}$ over the optimal, $K_{opt}$. 
\end{enumerate}

Showing $\tfrac{\partial^2S_{AMG}}{\partial K^2} > 0$ is equivalent to showing $K^3\mathcal{G}'(K) > K^2\mathcal{G}(K)$. 
Plugging in and expanding, this is equivalent to proving
\begin{align*}
\frac{\zeta K}{\xi} & \leq ME \left[ \frac{\alpha(\alpha-1)}{\ln(2)}\log_2\left(\frac{P}{K}\right)^{\alpha-2} + 3\alpha
	\log_2\left(\frac{P}{K}\right)^{\alpha-1} + 2\ln(2)\log_2\left(\frac{P}{K}\right)^{\alpha}\right].
\end{align*}
Due to the constraints that $K \leq \min\{P, ME\}$ and the fact that $\zeta < \xi$, it is sufficient to prove
the interior term on the right is $\geq 1$. Let us make the change of variables $K = P/2^\ell$, for $\ell > 0$,
encompassing the range $K \in (0,P)$. Then we need to show
\begin{align*}
\mathcal{H}_\alpha(\ell) := \frac{\alpha(\alpha-1)}{\ln(2)}\ell^{\alpha-2} + 3\alpha\ell^{\alpha-1} + 2\ln(2)\ell^\alpha \geq 1.
\end{align*}
for $\ell > 0$ and $\alpha \in (1,2)$. A second-derivative test confirms that $\mathcal{H}_\alpha(\ell)$ is concave up
for $\alpha\in(1,2)$ and further algebra confirms a single critical point of $\mathcal{H}_\alpha$ for $\alpha\in(1,2)$ at 
\begin{align*}
\ell_0 = \frac{3 - 3\alpha + \sqrt{\alpha^2+6\alpha-7}}{4 \ln(2)},
\end{align*}
where the minimum of $\mathcal{H}_\alpha(\ell)$ occurs at $\ell_0 \in(0,0.248)$ for $\alpha \in(1,2)$. Plugging $\ell_0$ into
$\mathcal{H}_\alpha(\ell)$ and plotting as a function of $\alpha\in(1,2)$ confirms that 
$\mathcal{H}_\alpha(\ell) \geq \mathcal{H}_\alpha(\ell_0) > 1$.

Due to the additional practical constraint that $K$ must be integer valued and divide the number of processors,
there is strong evidence that for $\alpha \in (1,2)$,
\begin{align*}
K_{opt} = \min\left\{ P, ME, \frac{PC_P- NC_TME}{C_aNC_T}\right\}.
\end{align*}
\end{proof}

\section*{Acknowledgments}
This material is based upon work supported by the Department of Energy, National Nuclear Security Administration, under Award Number DE-NA0002376.

\bibliographystyle{ans_js}                                                                           \bibliography{bibliography}

\end{document}